\documentclass[a4paper,USenglish,numberwithinsect]{lipics-report}
\usepackage[table]{xcolor}
\usepackage{amsthm}
\usepackage{wrapfig}
\usepackage{cancel}
\usepackage{caption}
\usepackage{amsmath}

\usepackage{scalerel}
\DeclareMathOperator*{\Bigcdot}{\scalerel*{\cdot}{\bigodot}}
\usepackage{amssymb}
\usepackage[version=4]{mhchem}
\usepackage{mathtools}
\usepackage{eufrak}
\usepackage{tikz}
\usepackage{pifont}
\usetikzlibrary{
    arrows,
    calc,
    positioning,
    intersections,
    shapes,
    snakes,
    automata,
    backgrounds,
    petri
}

\tikzset{cross/.style={cross out, draw=black, minimum size=2*(#1-\pgflinewidth), inner sep=0pt, outer sep=0pt},
cross/.default={1pt}}

\newtheorem{claim}{Claim}[section]

\newtheorem{problem}{Problem}
\newcommand\tab[1][1cm]{\hspace*{#1}}

\renewcommand{\,}{\hspace{0.1cm}}
\newcommand{\Zd}{\ensuremath{\mathsf{PN}}}
\newcommand{\Zc}[1]{\textsf{{#1}-PN}}
\newcommand{\Zb}[2]{\textsf{{#1}-H{#2}PN}}
\newcommand{\Za}[1]{\textsf{H{#1}PN}}
\newcommand{\Ze}[2]{{{#1}-{#2}}}
\newcommand{\N}{\mathbb N}
\newcommand{\Arrow}[1]{\ce{->[#1]}}
\newcommand{\Rng}{\ensuremath\mathsf{Range}}

\newcommand{\term}{\textsc{Term}}
\newcommand{\cover}{\textsc{Cover}}
\newcommand{\reach}{\textsc{Reach}}
\newcommand{\dlf}{\textsc{DLFree}}
\newcommand{\Markings}[1]{\ensuremath{\mathsf{Markings}(#1)}}
\newcommand{\Deadlock}[1]{\ensuremath{\mathsf{Deadlock}(#1)}}
\newcommand{\NotFire}[2]{\ensuremath{\mathsf{NotFire}_{#1}(#2)}}
\newcommand{\cN}{\mathcal{N}}
\newcommand{\cM}{\mathcal{M}}

\title{On Petri Nets with Hierarchical Special Arcs}
\author[1]{S Akshay}
\affil[1]{Department of Computer Science and Engineering, IIT Bombay}
\author[1]{Supratik Chakraborty}
\author[2]{Ankush Das}
\affil[2]{Computer Science Department, Carnegie Mellon University}
\author[1]{\\Vishal Jagannath}
\author[1]{Sai Sandeep}
\authorrunning{S. Akshay, S. Chakraborty, A. Das, V. Jagannath, S. Sandeep}


\begin{document}
\maketitle

\begin{abstract}
  We investigate the decidability of termination, reachability,
  coverability and deadlock-freeness of Petri nets endowed with a
  hierarchy on places, and with inhibitor arcs, reset arcs and
  transfer arcs that respect this hierarchy.  We also investigate what
  happens when we have a mix of these special arcs, some of which
  respect the hierarchy, while others do not. We settle the
  decidability status of the above four problems for all combinations
  of hierarchy, inhibitor, reset and transfer arcs, except the
  termination problem for two combinations.  For both these
  combinations, we show that the termination problem is as hard as
  deciding positivity for linear recurrent sequences --- a long-standing open problem.
\end{abstract}

\section{Introduction}
\label{sec:intro}
Petri nets are an important and versatile mathematical modeling
formalism for distributed and concurrent systems. Thanks to their
intuitive visual representation, precise execution semantics,
well-developed mathematical theory and availability of tools for
reasoning about them, Petri nets are used for modeling in varied contexts,
viz. computational, chemical, biological, workflow-related etc. Several extensions to Petri nets have been proposed in the literature
to augment their modeling power.  From a theoretical perspective,
these provide rich and interesting models of computation that warrant
investigation of their expressive powers, and decidability and/or
complexity of various decision problems.  From a practitioner's
perspective, they enable new classes of systems to be modeled and reasoned about.

In this paper, we focus on an important class of extensions proposed earlier for Petri nets, pertaining to the addition of three types of \emph{special} arcs, namely \emph{inhibitor}, \emph{reset} and \emph{transfer} arcs from places to transitions. We investigate how different combinations of these extensions affect the decidability of four key decision problems: reachability (whether a given marking can be reached), coverability, (whether a marking can be covered), termination, (whether the net has an infinite run) and deadlock-freeness (whether the net can get to a marking where no transition is fireable). To start with, an inhibitor-arc effectively models a zero test, and hence with two inhibitor arcs one can model two-counter machines, leading to undecidability of all of the above decision problems. However, Reinhardt~\cite{reinhardt} showed that if we
impose a \emph{hierarchy} among places with inhibitor arcs (a
single inhibitor arc being a sub-case), we recover decidability
of reachability. Recently, Bonnet~\cite{bonnet} gave a simplified proof of this using techniques of Leroux~\cite{leroux} and also showed that termination and coverability are decidable for Petri nets with hierarchical
inhibitor arcs. With reset arcs (which remove all tokens from a pre-place) and
transfer arcs (which transfer all tokens from a pre-place to a
post-place), reachability and deadlock-freeness are known to be
undecidable~\cite{transfer}, athough termination and coverability are
decidable~\cite{finkel}.

Our interest in this paper lies in asking what happens when
\emph{hierarchy} is introduced among all combinations of special arcs. Thus, we
specify a hierarchy, or total ordering, among the places and say that
the special arcs respect the hierarchy if whenever there is a special
arc from a place $p$ to a transition $t$, there are also special arcs
from every place lower than $p$ in the hierarchy to $t$. The study of
Petri nets extended with hierarchical and non-hierarchical special
arcs provides a generic framework that subsumes several existing
questions and throws up new ones. While some of these classes and
questions have been studied earlier, there are still several classes
where nothing seems to be known about the decidability of the above decision
problems.

Decidability
of reachability for Petri nets with hierarchical inhibitor arcs was
shown in ~\cite{reinhardt,Bonnet11}, while decidability of
termination, coverability and boundedness was shown
in~\cite{bonnet}. Further, in~\cite{AG11} it was shown that Petri nets
with hierarchical zero tests are equivalent to Petri nets with a stack
encoding restricted context-free languages. Finally a specific
subclass, namely Petri nets with a single inhibitor arc, has received a lot
of attention, with results showing decidability of boundededness and
termination~\cite{FS10}, place-boundedness~\cite{BFLZ10}, and LTL
model checking~\cite{BFJZ12}. However, in~\cite{BFJZ12}, the authors
remark that it would not be easy to extend their technique for the
last two problems to handle hierarchical arcs. To the best of our
knowledge, none of the earlier papers address the mixing of reset and
transfer arcs within the hierarchy of inhibitor arcs, leaving several
interesting questions unanswered.
Our primary goal in this paper is to comprehensively fill these gaps.
Before we delve into theoretical investigations of these models, we
present two examples that illustrate why these models are
interesting from a practical point of view too.

Our first example is a prioritized job-shop environment in which work
stations with possibly different resources are available for servicing
jobs.  Each job comes with a priority and with a requirement of the
count of resources it needs.  For simplicity, assume that all
resources are identical, and that there is at most one job with any
given priority.  A work station can service multiple jobs
simultaneously subject to availability of resources; however, a job
cannot be split across multiple work stations.  Additionally, we
require that a job with a lower priority must not be scheduled on any
work station as long as a job with higher priority is waiting to be
scheduled.  Once a job gets done, it can either terminate or generate
additional jobs with different priorities based on some rules.  An
example of such a rule could be that a job with prioriy $k$ and
resource requirement $m$ can only generate a new job with priority
$\le k$ and resource requirement $\le m$.
Given such a system, there are several interesting questions one might
ask. For example, can too many jobs (above a specified threshold) of
the lowest priority be left waiting for a work station?  Or, can
the system reach a deadlocked state from where no progress can be
made?  A possible approach to answering these questions is to model
the system as a Petri net with appropriate extensions, and reduce the
questions to decision problems (such as coverability or
deadlockfree-ness) for the corresponding nets.  Indeed, it can be shown
that the prioritized job-shop environment can be modeled as a Petri net
with hierarchical inhibitor arcs and additional transfer/reset arcs
that do not necessarily respect the hierarchy.

Our second example builds on work reported in the literature on
modeling integer programs with loops using Petri nets~\cite{AAGM12}.
Questions pertaining to termination of such programs can be reduced to
decision problems (termination or deadlockfree-ness) of the
corresponding Petri net model.  In Section~\ref{sec:skolem}, we
describe a new reduction of the termination question for integer
linear loop programs to the termination problem for Petri nets with
hierarchical inhibitor and transfer arcs. This is one of the main
technical contributions of this paper, and underlines the importance
of studying decision problems for these extensions of Petri nets.

Our other main contribution is a comprehensive investigation into
Petri nets extended with a mix of these special arcs, some of which
respect the hierarchy, while others do not. We settle the decidability
status of the four decisions problems for \emph{all} combinations of
hierarchy, inhibitor, reset and transfer arcs, except the termination
problem for two combinations.  For these cases, we show a reduction
from the positivity problem~\cite{positivity1,positivity2}, a
long-standing open problem on linear recurrences. We summarize these
results in Section~\ref{sec:problem}, after introducing appropriate
notations in Section~\ref{sec:prelim}. Interestingly, several of the
results use distinct constructions and proof techniques, as detailed
in Sections~\ref{sec:hirpn}--\ref{sec:transf}. For the sake of clarity and ease of reading, we provide proof sketches of most of the results in the main body of the article and defer the detailed formal proofs to the appendix.

\section{Preliminaries}
\label{sec:prelim}
We begin by recalling some key definitions and fixing notations.

\begin{definition}
A Petri net, denoted \Zd, is defined as $(P,T,F,M_0)$, where $P$ is a
set of \emph{places}, $T$ is a set of \emph{transitions},
$M_0:P\rightarrow\N$ is the \emph{initial marking} and $F:(P\times
T)\cup(T\times P)\rightarrow \N $ is the \emph{flow relation}.
\end{definition}
Consider a petri net $N=(P,T,F,M_0)$.  For every $x\in P\cup T$, we define $Pre(x)=\{y\in P\cup T \mid F(y,x)>0\}, Post(x)=\{y\in P\cup T \mid F(x,y)>0\}$. For every $t\in T$, we use the following terminology: every $p\in
Pre(t)$ is a \emph{pre-place} of $t$, every $q\in Post(t)$ is a
\emph{post-place} of $t$, every arc $(p, t)$ such that $F(p,t)>0$ is a
\emph{pre-arc} of $t$, and every arc $(t, p)$ such that $F(t,p)>0$ is a
\emph{post-arc} of $t$.

A \emph{marking} $M: P\rightarrow \mathbb{N}$ is a function from the
set of places to non-negative integers. We say that a transition $t$
is \emph{firable} at marking $M$, denoted by $M\xrightarrow{t}$, if
$\forall p\in Pre(t),M(p)\geq F(p,t)$. If $t$ is firable at $M_1$, we
say that firing $t$ gives the marking $M_2$, where $\forall p\in P,
M_2(p)=M_1(p)-F(p,t)+F(t,p)$.  This is also denoted as
$M_1\xrightarrow{t} M_2$. We define the sequence of transitions
$\rho=t_1t_2t_3...t_n$ to be a \emph{run} from marking $M_0$, if there
exist markings $M_1,M_2,...,M_n$, such that for all $i$, $t_i$ is
firable at $M_{i-1}$ and $M_{i-1}\xrightarrow{t_i}M_i$.
Finally, we abuse notation and use $\leq$ to denote the component-wise
ordering over markings.  Thus, $M_1\leq M_2$ iff $\forall p\in
P,M_1(p)\leq M_2(p)$. A detailed account on Petri nets can be found in~\cite{murata}.

We now define some classical decision problems in the study of Petri nets. 
\begin{definition} Given a Petri net $N = (P, T, F, M_0)$,
\begin{itemize}
    \item \textsf{Termination} (or {\term}): Does there exist an
      infinite run from marking $M_0$?
    \item \textsf{Reachability} (or {\reach}): Given a marking $M$, is there
      a run from $M_0$ which reaches $M$?
    \item \textsf{Coverability} (or {\cover}): Given a marking $M$, is there     a  marking $M'\geq M$ which is reachable from $M_0$?
    \item \textsf{Deadlock-freeness} (or {\dlf}):  Does there exist a marking $M$
      reachable from $M_0$, such that no transition is firable at $M$?
\end{itemize}
\end{definition}
Since Petri nets are well-structured transition systems (WSTS), the
decidability of coverability and termination for Petri nets follows
from the corresponding results for WSTS~\cite{finkel}.  The
decidability of reachability was shown in~\cite{kosaraju}.
Subsequently, there have been several alternative proofs of the same
result, viz.~\cite{leroux}. Finally, since deadlockfreeness reduces to
reachability in Petri nets~\cite{dfresult}, all the four decision
problems are decidable for Petri nets.  In the remainder of the paper,
we concern ourselves with these decision problems for Petri nets extended with the following special arcs:

\begin{itemize}
    \item An \textsf{Inhibitor arc} from place $p$ to
      transition $t$ signifies $t$ is firable only if $p$ has zero
      tokens.
    \item A \textsf{Reset arc} from place $p$ to transition
      $t$ signifies that $p$ contains zero tokens after $t$ fires.
    \item A \textsf{Transfer arc} from place $p_1$ through transition $t$ to place $p_2$ signifies that on firing transition $t$, all tokens from $p_1$ get transferred to $p_2$.
\end{itemize}

For Petri nets with special arcs, we redefine the flow relation as
$F:(P\times T)\cup (T\times P)\rightarrow
\mathbb{N}\cup\{I,R\}\cup\{S_p \mid p\in P\}$, where $F(p,t)=I$
(resp. $F(p,t)=R$) signifies the presence of an inhibitor arc
(resp. reset arc) from place $p$ to transition $t$. Similarly, if $F(p,t)=S_{p'}$, then there is a transfer arc from place
$p$ to place $p'$ through transition $t$.

\section{Problem statements and main results}
\label{sec:problem}
We now formally define the various extensions of Petri nets studied in
this paper.  We also briefly review the current status (with respect
to decidability) of the four decision problems for these extensions of
Petri nets, and summarize our contributions.

We use \Zd{} to denote standard Petri nets, and \Zc{I}, \Zc{R}, \Zc{T}
to denote Petri nets with inhibitor, reset and transfer arcs,
respectively.  The following definition subsumes several additional
extensions studied in this paper.

\begin{definition}
A Petri net with \emph{hierarchical special arcs} is defined to be a $5$-tuple $(P, T, F, \sqsubseteq,
M_0)$, where $P$ is a set of places, $T$ is a set of transitions,
$\sqsubseteq$ is a total ordering over $P$ encoding the hierarchy,
$M_0: P \rightarrow \N$ is the initial marking, and $F: (P \times T)
\cup (T \times P) ~\rightarrow~ \N \cup\{I,R\}\cup\{S_p \mid p\in P\}$
is a flow relation satisfying
\begin{itemize}
\item $\forall (t,p)\in T\times P,~~F(t,p)\in \N$, and 
\item $\forall(p,t)\in P\times T,~~F(p,t)\not\in\N \implies \left(\forall q\sqsubseteq p,~F(q,t)\not\in\N\right)$
\end{itemize}
\end{definition}

Thus, all arcs (or edges) from transitions to places are as in
standard Petri nets.  However, we may have special arcs from places to
transitions. These can be inhibitor arcs ($F(p, t) = I$), reset arcs
($F(p, t) = R$), or transfer arcs ($F(p,t) = S_{p'}$, where $p$ and
$p'$ are places in the Petri net).  Note that all special arcs respect
the hierarchy specified by $\sqsubseteq$.  In other words, if there is
a special arc from a place $p$ to a transition $t$, there must also be
special arcs from every place $p'$ to $t$, where $p' \sqsubseteq p$.

Depending on the subset of special arcs that are present, we can define
sub-classes of Petri nets with hierarchical special arcs as follows.
In the following, ${\Rng}(F)$ denotes the range of the flow relation $F$.
\begin{definition}
  \label{def:pn-hsa}
  The class of Petri nets with hierarchical special arcs, where
  ${\Rng}(F) \setminus \N$ is a subset of $\{I\}, \{T\}$ or $\{R\}$ is
  called \Za{I}, \Za{T} or \Za{R}
  respectively.  Similarly, it is called \Za{IT},
  \Za{IR} or \Za{TR} if ${\Rng}(F) \setminus \N$ is
  a subset of $\{I, T\}, \{I, R\}$ or $\{T, R\}$ respectively.
  Finally, if ${\Rng}(F) \setminus \N$ is a subset of $\{I, R,
  T\}$, we call the corresponding class of nets \Za{IRT}.
\end{definition}

We also study generalizations, in which
extra inhibitor, reset and/or transfer arcs that do not respect the
hierarchy specified by $\sqsubseteq$, are added to Petri nets with
hierarchical special arcs.

\begin{definition}
  Let $\cN$ be a class of Petri nets with hierarchical special
  arcs as in Definition~\ref{def:pn-hsa}, and let $\cM$ be a
  subset of $\{I, T, R\}$.  We use
  \Ze{$\cM$}{$\cN$} to denote the class of nets
  obtained by adding unrestricted special arcs of type $\cM$
  to an underlying net in the class $\cN$.
\end{definition}

For example, $\Zb{R}{I}$ is the class of Petri nets with hierarchical
inhibitor arcs extended with reset arcs that need not respect the
hierarchy. Clearly, if the special arcs in every net $N\in\cN$ are
from $\cM$, the class \Ze{$\cM$}{$\cN$} is simply the class of Petri
nets with unrestricted (no hierarchy) arcs of type $\mathcal{M}$.
Hence we avoid discussing such extensions in the remainder of the
paper.

As we show later, all four decision problems of interest to us are either undecidable or not known to be decidable for \Za{IRT}. A slightly constrained version of \Za{IRT}, however, turns out to be much better behaved, motivating the following definition. 
\begin{definition}
The sub-class \Za{IRcT} is defined to be  \Za{IRT}
with the added restriction that $\forall (p, t, p')\in P\times T\times
P, ~F(p,t)=S_{p'}\implies \left(F(p',t)\in
\N\right)$.
\end{definition}

\subsection{Status of decision problems and our contributions}
  Table~\ref{summary} summarizes the decidability status of the four
  decision problems for some classes of Petri net extensions.  A
  \ding{51} denotes decidability of the corresponding problem, while
  \ding{55} denotes undecidability of the problem.  The shaded cells
  present results (and corresponding citations) already known prior to
  the current work, while the unshaded cells show results (and
  corresponding theorems) arising from this paper.  Note that the
  table doesn't list all extensions of Petri nets that were defined
  above.  This has been done deliberately and carefully to improve
  readability.  Specifically, for every Petri net extension that is
  not represented in the table, e.g., \Zb{R}{IT}, the status of all
  four decision problems are inferable from others shown in the table. These are explicitly listed out in Appendix~\ref{sec:app-summary}, where we also depict the relative expressiveness of these classes.  Thus, our work comprehensively addresses the four decision problems for all classes of Petri net extensions considered above. 

\begin{table}[t]
\begin{tabular}{|c|c|c|c|c|}
    \hline
     & {\term} & {\cover} & {\reach} & {\dlf} \\ 
    \hline
    \Zd & \cellcolor{lightgray} \ding{51} \cite{finkel} & \cellcolor{lightgray} \ding{51} \cite{finkel} & \cellcolor{lightgray} \ding{51} \cite{mayr,leroux} & \cellcolor{lightgray} \ding{51} \cite{dfresult,hack}\\
    \hline
    \Zc{R/T} & \cellcolor{lightgray} \ding{51} \cite{finkel} & \cellcolor{lightgray} \ding{51} \cite{finkel} & \cellcolor{lightgray} \ding{55} \cite{transfer} & \cellcolor{lightgray} \ding{55} [Red. from~\cite{transfer}]\\
    \Zc{I} & \cellcolor{lightgray} \ding{55} \cite{inhibitor} & \cellcolor{lightgray} \ding{55} \cite{inhibitor} & \cellcolor{lightgray} \ding{55} \cite{inhibitor} & \cellcolor{lightgray} \ding{55} \cite{inhibitor}\\
    \hline
    \Za{I} & \cellcolor{lightgray} \ding{51} \cite{reinhardt,bonnet} & \cellcolor{lightgray} \ding{51} \cite{reinhardt,bonnet} & \cellcolor{lightgray} \ding{51} \cite{reinhardt,bonnet} & \ding{51} [Thm~\ref{thm:hirpn-dlf}]\\
    \Za{T} & \cellcolor{lightgray} \ding{51} \cite{finkel} & \cellcolor{lightgray} \ding{51} \cite{finkel} & \ding{55} [Thm~\ref{thm:htpn-reach-dlf}] & \ding{55} [Thm~\ref{thm:htpn-reach-dlf}]\\
    \Za{IR} & \ding{51} [Thm~\ref{thm:bisim-k}] & \ding{51} [Thm~\ref{thm:bisim-k}] & \ding{51} [Thm~\ref{thm:bisim-k}] & \ding{51} [Thm~\ref{thm:hirpn-dlf}]\\
    \Za{IT} & Positivity-Hard [Thm~\ref{thm:skolem-hard}] & \ding{55} [Cor. \ref{hitpncover}] & \ding{55} [Thm~\ref{thm:htpn-reach-dlf}] & \ding{55} [Thm~\ref{thm:htpn-reach-dlf}]\\
    \Za{IRcT} & \ding{51} [Thm~\ref{thm:bisim-k}] & \ding{51} [Thm~\ref{thm:bisim-k}]
 & \ding{51} [Thm~\ref{thm:bisim-k}] & \ding{51} [Thm~\ref{thm:hirpn-dlf}]\\
    \hline
    \Zb{R}{I} &\ding{51}[Thm~\ref{thm:rhipn-term}] & \ding{55}[Thm~\ref{cover-thm}] & \cellcolor{lightgray} \ding{55}[\cite{transfer}, App.~\ref{1I1RPN}] & \cellcolor{lightgray}\ding{55}[Red.frm~\cite{transfer},~\ref{1I1RPN}]\\
    \Zb{T}{I} & Positivity-Hard [Thm~\ref{thm:skolem-hard}] & \ding{55}[Thm~\ref{cover-thm}] & \cellcolor{lightgray} \ding{55}[\cite{transfer}, App.~\ref{1I1RPN}]  & \cellcolor{lightgray}\ding{55}[Red.frm~\cite{transfer},~\ref{1I1RPN}]\\
    \Zb{R}{IR} &\ding{51}[Thm~\ref{thm:rhipn-term}, Thm~\ref{thm:bisim-k}] & \ding{55}[Thm~\ref{cover-thm}] & \cellcolor{lightgray} \ding{55}[\cite{transfer}, App.~\ref{1I1RPN}]  & \cellcolor{lightgray}\ding{55}[Red.frm~\cite{transfer},~\ref{1I1RPN}]\\
    \hline
\end{tabular}
\captionof{table}{Summary of key results; results for all other extensions are subsumed by these results}\label{summary}
\end{table}

Interestingly, several of the results use distinct constructions and proof techniques. We now point out the salient features of our six main results.

\begin{itemize}
\item We \emph{include reset arcs in the hierarchy} of inhibitors in \Za{I} in Section~\ref{sec:hirpn}. In Theorem~\ref{thm:bisim-k}, we show that we can model reset arcs by inhibitors, while crucially preserving hierarchy. This immediately gives decidability of all problems except \dlf. As the reduction may introduce deadlocks, we need a different proof for \dlf, which we show in our second main and more technically involved result in Theorem~\ref{thm:hirpn-dlf}. Note that this immediately also proves decidability of deadlock-freeness for \Za{I}, which to the best of our knowledge was not known before.
\item We \emph{add reset arcs outside the hierarchy} of inhibitor arcs in Section~\ref{sec:rhipn}. Somewhat counter-intuitively, this class does not contain \Za{I} and is incomparable to it, since here all inhibitor arcs must follow the hierarchy, while in \Za{I} some of the inhibitor arcs can be replaced by resets. Using a new and surprisingly simple construction of an extended finite reachability tree (FRT) which keeps track of the hierarchical inhibitor information  and modifies the subsumption condition, in Theorem~\ref{thm:rhipn-term}, we show that termination is decidable. This result has many consequences. In particular, it implies an arguably simple proof for the very special case of a single inhibitor arc which was solved in~\cite{FS10} (using a different method of extending FRTs). In Theorem~\ref{cover-thm}, we use a two counter machine reduction to show that coverability is undecidable even with 2 reset arcs  and an inhibitor arc in the absence of hierarchy.
\item Finally, we \emph{consider transfer arcs in and outside the hierarchy} in Section~\ref{sec:transf}. 
In Theorem~\ref{thm:htpn-reach-dlf}, we show that, unlike for reset arcs, including transfer arcs in the hierarchy of inhibitors does not give us decidability. For both \Za{IT} and \Zb{T}{I}, while coverability, reachability and  deadlock-freeness are undecidable, we are unable to show such a result for termination. Instead, in Theorem~\ref{thm:skolem-hard}, we show that we can reduce a long-standing open problem on linear recurrence sequences to this problem. 
\end{itemize}

\section{Adding Reset Arcs with Hierarchy to \Za{I}}
\label{sec:hirpn}
In this section, we extend hierarchical inhibitor nets~\cite{reinhardt} with reset arcs respecting the hierarchy. Subsection~\ref{sec:HIRPN} presents a reduction from \Za{IR} to \Za{I} that settles the decidability of termination, coverability and reachability for \Za{IR}.  Unfortunately, this reduction does not work
for deadlock-freeness since it introduces new deadlocked markings (as explained later). We therefore present a separate reduction from deadlock-freeness to reachability for \Za{IR} in subsection~\ref{sec:deadlock}. These two results establish the decidability of all four decision problems for \Za{IR}.

\subsection{Reduction from \Za{IR} to \Za{I}}\label{sec:HIRPN}
In this subsection, we present a reduction from Inhibitor-Reset Petri Nets to Inhibitor Petri Nets, such that the hierarchy is preserved. This reduction holds for Termination, Coverability and Reachability. Thus, as a consequence of Reachability in \Za{I}, it follows that Reachability in \Za{IR} is also decidable. In particular, reachability in Petri Nets with 1 Reset Arc is decidable.

Let \Za{IR}$_k$ be the sub-class of Petri nets in \Za{IR} with at most
$k$ transitions having one or more reset pre-arcs.  We first show that
termination, reachability and coverability for \Za{IR}$_k$ can be
reduced to the corresponding problems for \Za{IR}$_{k-1}$, for all $k
> 0$.  This effectively reduces these problems for \Za{IR} to the
corresponding problems for \Za{IR}$_0$ (or \Za{I}), which are known to
be decidable~\cite{reinhardt,bonnet}. In the following, we use
$\Markings{N}$ to denote the set of all markings of a net $N$.

\begin{lemma} \label{lem:hirpnk-to-hirpnk-1}
  For every net $N$ in \Za{IR}$_k$, there is a net $N'$ in
  \Za{IR}$_{k-1}$ and a mapping $f : \Markings{N} \rightarrow
  \Markings{N'}$ that satisfy the following:
  \begin{itemize}
  \item For every $M_1, M_2 \in \Markings{N}$ such that $M_2$ is
    reachable from $M_1$ in $N$, the marking $f(M_2)$ is reachable from $f(M_1)$
    in $N'$.
  \item For every $M_1', M_2' \in \Markings{N'}$ such that
    $M_1' = f(M_1)$, $M_2' = f(M_2)$ and $M_2'$ is reachable
    from $M_1'$ in $N'$, the marking $M_2$ is reachable from $M_1$ in $N$. 
  \end{itemize}
\end{lemma}

\begin{figure}[t]
\begin{tikzpicture}[scale=0.75]
\hspace{0pt}\raisebox{-1em}
{
\draw (0,13.875) circle (0.3cm) node{\scriptsize $p^S$};
\draw (1,13.875) circle (0.3cm) node{\scriptsize $p^I$};
\draw (2,13.875) circle (0.3cm) node{\scriptsize $p^{R}$};
\draw (1,11.625) circle (0.3cm) node{\scriptsize $p$};
\draw [fill=black] (0.65,12.675) node[left]{\scriptsize $t$} rectangle (1.35,12.825);
\draw[-latex,thick] (0,13.575) -- (0.65,12.825);
\draw[-o,thick] (1,13.575) -- (1,12.825);
\draw[-latex,thick] (2,13.575) -- node[right]{\scriptsize Reset} (1.35,12.825);
\draw[-latex,thick] (1,12.675) -- (1,11.925);
\draw (2,12.75) rectangle node{\scriptsize \textit{Rest of Net}} (4,11.625);
}

\raisebox{3em}
{
\path[draw=black,solid,line width=2mm,fill=black,
preaction={-triangle 90,thick,draw,shorten >=-1mm}
] (4.6, 10.875) -- (6.3, 10.875);
}

\hspace{180pt}\raisebox{9.3em}
{
\draw (0,10.125) circle (0.3cm) node{\scriptsize $p^S$};
\draw (2,10.125) circle (0.3cm) node{\scriptsize $p^*$};
\draw (1,7.875) circle (0.3cm) node{\scriptsize $p_t^*$};
\draw [fill=black] (0.65,8.925) node[left]{\scriptsize $t^S$} rectangle (1.35,9.075);
\draw[-latex,thick] (0,9.825) -- (0.8,9.075);
\draw[-latex,thick] (1.9,9.825) -- (1.1,9.075);
\draw[-latex,thick] (1,8.925) -- (1,8.175);
\draw (3,9.375) rectangle node{\scriptsize \textit{Rest of Net}} (5,8.25);
\draw[-latex,thick,dotted] (2.4,10.125) to[out=-20, in=100] (3.8, 9.15);
\draw[-latex,thick,dotted] (3.2,9.15) to[out=170, in=-60] (2,9.825);

\draw (2,6.75) circle (0.3cm) node{\scriptsize $p^{R}$};
\draw (0,6.75) circle (0.3cm) node{\scriptsize $p^I$};
\draw (1,4.5) circle (0.3cm) node{\scriptsize $p$};
\draw [fill=black] (3.5,6.6375) node[below]{\scriptsize $t^R$} rectangle (3.7,7.1625);
\draw [-latex,thick] (2.3,6.75) -- (3.5,6.75);
\draw [-latex,thick] (1,7.575) .. controls (1.5,7.275) and (2.5,7.125) .. (3.5,7.05);
\draw [-latex,thick] (3.7,7.05) .. controls (4.5,7.5) and (2.5,7.725) .. (1.3,7.875);
\draw [fill=black] (1.35,5.7) node[below right]{\scriptsize $t^I$} rectangle (0.65,5.55);
\draw[-o,thick] (0,6.45) -- (0.65,5.7);
\draw[-latex,thick] (1,7.575) -- (1,5.7);
\draw[-o,thick] (2,6.45) -- (1.35,5.7);
\draw[-latex,thick] (1,5.55) -- (1,4.8);
\draw[-latex,thick] (0.8,5.55) .. controls (-2,3.375) and (-3,13.5) .. (2,10.425);
}
\end{tikzpicture}
\vspace{-11.5em}
\caption{Transformation from $N \in$ \Za{IR}$_k$ (left) to $N' \in$ \Za{IR}$_{k-1}$ (right)}
\label{fig:transform-hirpn}
\end{figure}
\noindent \emph{Proof sketch:}
To see how $N'$ is constructed, consider an arbitrary transition, say
$t$, in $N$ with one or more reset pre-arcs.  We replace $t$ by a
gadget in $N'$ with no reset arcs, as shown in
Figure~\ref{fig:transform-hirpn}.  The gadget has two new places
labeled $p^*$ and $p_t^*$, with every transition in ``Rest of Net''
having a simple pre-arc from and a post-arc to $p^*$, as shown by the
dotted arrows in Figure~\ref{fig:transform-hirpn}.  The gadget also
has a new transition $t^S$ with simple pre-arcs from $p^*$ and from
every place $p^S$ that has a simple arc to $t$ in $N$.  It also has a
new transition labeled $t^R$ for every reset arc from a place $p^R$ to
$t$ in $N$.  Thus, if there are $n$ reset pre-arcs of $t$ in $N$, the
gadget will have $n$ transitions $t^R_1, \ldots t^R_n$.  As shown in
Figure~\ref{fig:transform-hirpn}, each such $t^R_i$ has simple
pre-arcs from $p^R_i$ and $p_t^*$ and a post-arc to $p_t^*$.  Finally,
the gadget has a new transition labeled $t^I$ with a simple pre-arc
from $p_t^*$ and inhibitor pre-arcs from all places $p^I$ that have
inhibitor arcs to $t$ in $N$.

The ordering $\sqsubseteq'$ of places in $N'$ is obtained by extending
the ordering $\sqsubseteq$ of $N$ as follows: for each place $p$ in
$N$, we have $p \sqsubseteq' p_t^* \sqsubseteq' p^*$.  Clearly, $N'
\in \Za{IR}_{k-1}$, since it has one less transition (i.e. $t$) with
reset pre-arcs compared to $N$.  It is easy to check that if the reset
and inhibitor arcs in $N$ respect $\sqsubseteq$, then the reset and
inhibitor arcs in $N'$ respect $\sqsubseteq'$.

The mapping function $f: \Markings{N} \rightarrow \Markings{M'}$ is
defined as follows: for every place $p$ in $N$, $f(M)(p) = M(p)$ if
$p$ is in $N$; otherwise, $f(M)(p^*) = 1$ and $f(M)(p_t^*) = 0$. The initial marking of $N'$ is given by $f(M_0)$, where $M_0$ is the
initial marking of $N$.  Given a run in $N$, it is now easy to see
that every occurrence of $t$ in the run can be replaced by the
sequence $t^S(t^R)^*t^I$ (the $t^R$ transitions fire until the
corresponding place $p^R$ is emptied) and vice-versa.  Further details
of the construction are given in Appendix~\ref{app:hirpn_bisimulation},
where it is also shown that $N$ can reach $M_2$ from $M_1$ iff $N'$ can
reach $f(M_2)$ from $f(M_1)$.  

In fact, the above construction can be easily adapted for \Za{IRcT} as
well.  Specifically, if we have a transfer arc from place $p_x$ to
place $p_y$ through $t$, we add a new transtion $t_{x,y}^T$ with
simple pre-arcs from $p_t^*$ and $p_x$, and with simple post-arcs to
$p_t^*$ and $p_y$ to the gadget shown in
Figure~\ref{fig:transform-hirpn}.  Furthermore, we add an inhibitor
arc from $p_x$ to $t^I$, like the arc from $p^R$ to $t^I$ in
Figure~\ref{fig:transform-hirpn}.  This allows us to obtain a net in
\Za{IRcT} with at least one less transition with reset pre-arcs or
transfer arcs, such that the reachability guarantees in
Lemma~\ref{lem:hirpnk-to-hirpnk-1} hold.  This immediately gives the
following result.
\begin{theorem}\label{thm:bisim-k}
  Termination, reachability and coverability for \Za{IR} and \Za{IRcT} are decidable.
\end{theorem}
The proof follows by repeatedly applying
Lemma~\ref{lem:hirpnk-to-hirpnk-1} to reduce the decision problems to
those for \Za{IR}$_0$ (or \Za{I}), and from the decidability of these
problems for \Za{I}.

We also note that the construction doesn't lift to \Za{IT}, because we
might have a transition, which transfers tokens into a place, say $p$,
and also has an inhibitor arc from $p$. In above construction, we would
have a transition $t^R$ adding tokens into place $p$ and an inhibitor arc
from place $p$ to $t^I$. Now, after transition $t^R$ is fired, it has
added tokens into $p$ and hence $t^I$ is not firable.


\subsection{Reducing Deadlock-freeness to Reachability in \Za{IR}}\label{sec:deadlock}
Note that the above reduction does not preserve  deadlock-freeness since during the reduction, we may introduce new deadlocked markings, in particular when transition $t^S$ fires and $p^I$ has tokens. Nevertheless, in this subsection, we show a different reduction from Deadlockfreeness in \Za{IR} to Reachability in \Za{IR}. Since reachability in \Za{IR} is decidable (as shown in Section \ref{sec:HIRPN}), this implies that Deadlockfreeness in \Za{IR} is decidable as well. The overall idea behind our reduction is to add transitions that check whether the net is deadlocked, and to put a token in a special place, say $p^*$, if this is indeed the case.  Note that for a net to be deadlocked, the firing of each of its transitions must be disabled. Intuitively, if $M$ denotes a marking of a net and if $T$ denotes the
set of transitions of the net, then $\Deadlock{M} = \bigwedge_{t_i\in
  T}\NotFire{i}{M}$, where $\Deadlock{M}$ is a predicate indicating if
the net is deadlocked in $M$, and $\NotFire{i}{M}$ is a formula
representing the enabled-ness of transition $t_i$ in $M$.

For a transition $t$ to be disabled, atleast one of its pre-places $p$
must fail the condition on that place for $t$ to fire. There are three
cases to consider here.
\begin{itemize}
	\item $F(p,t)\in \mathbb{N}$: For $t$ to be disabled, we must have $M(p) < F(p,t)$
	\item $F(p,t)=I$: For $t$ to be disabled, we must have $M(p) > 0$.
	\item $F(p,t)=R$: Place $p$ cannot disable $t$
\end{itemize}
Suppose we define $Exact_j(p) \equiv (M(p)=j)$ and $AtLeast(p) \equiv
(M(p)>0)$.  Clearly, $\NotFire{i}{M}= \bigvee_{(p,t)\in F} Check(p)$,
where $Check(p)= AtLeast(p)$ if $F(p,t) = I$, and $Check(p) =
\bigvee_{j < k}Exact_l(p)$ if $F(p,t) = k \in \N$.  The formula for
$\Deadlock{M}$ (in CNF above) can now be converted into DNF by
distributing conjunctions over disjunctions.  Given a \Za{IR} net, we
now transform the net, preserving hierarchy, so as to reduce checking
$\Deadlock{M}$ in DNF in the original net to a reachability problem in
the transformed \Za{IR} net.

Every conjunctive clause in the DNF of $\Deadlock{M}$ is a conjunction
of literals of the form $AtLeast(p)$ and $Exact_j(p)$. Let $S_C$ be
the set of all literals in a conjunctive clause C, and let $P$ be the
set of all places in the net. Define $B^C_i= \{p \in P \mid Exact_i(p)
\in S_C\}$ and $A^C= \{p\in P \mid AtLeast(p)\in S_C\} \setminus
\bigcup_{i\geq 1}B^C_i$.  We only need to consider conjunctive clauses
where the sets $B^C_i$ are pairwise disjoint (other clauses can never
be true). Similarly, we only need to consider conjunctive clauses
where $B^C_0$ and $A^C$ are disjoint. We add a transition for each
conjunctive clause that satisfies the above two properties. By
definition, $A^C$ and $B^C_i$ are disjoint for all $i\geq 1$. Thus,
the sets $A^C$ and $B^C_i (i\geq 0)$ are pairwise disjoint for every
conjunctive clause we consider.

\begin{figure}[t]
\centering
  \begin{tikzpicture}[scale=1.2]
\draw (0,5) circle (0.3cm) node{\scriptsize $p_2$};
\draw [fill=black] (-0.35,3.9) node[left]{\scriptsize $t_2$} rectangle (0.35,4.1);
\draw (0,3) circle (0.3cm) node{\scriptsize $p_{2*}$};
\draw [fill=black] (-0.35,1.9) node[left]{\scriptsize $t_C$} rectangle (4.35,2.1);
\draw [fill=black] (0.9,4.65) node[below]{\scriptsize $t_{2*}$} rectangle (1.1,5.35);
\draw [-o,thick] (-0.3,5) .. controls (-1.5,4) and (-1.5,3) .. (-0.2,2.1);
\draw [-latex,thick] (0.3,5) -- (0.9,5);
\draw [-latex,thick] (0,4.7) -- (0,4.1);
\draw [-latex,thick] (0,2.7) -- (0,2.1);
\draw [-latex,thick] (0,3.9) -- (0,3.3);
\draw [-latex,thick] (2,1.9) -- (2,1.3);
\draw (2,1) circle (0.3cm) node{\scriptsize $p_*$};

\draw (4,5) circle (0.3cm) node{\scriptsize $p_1$};
\draw [fill=black] (3.65,3.9) node[left]{\scriptsize $r_1$} rectangle (4.35,4.1);
\draw (4,3) circle (0.3cm) node{\scriptsize $p_{1*}$};
\draw [-o,thick] (3.7,5) .. controls (2.5,4) and (2.5,3) .. (3.8,2.1);
\draw [-latex,thick] (4,4.7) -- node[right]{\scriptsize i} (4,4.1);
\draw [-latex,thick] (4,2.7) -- (4,2.1);
\draw [-latex,thick] (4,3.9) -- (4,3.3);
\draw [dotted] (-1.5,0.5) rectangle (5.5,6);
\draw (6,1.2) circle (0.3cm) node{\scriptsize $p_C$};
\draw (6,3.2) circle (0.3cm) node{\scriptsize $p^{**}$};
\draw [fill=black] (5.65,2.1) rectangle (6.34,2.3) node[right]{\scriptsize $q_C$};
\draw [-latex,thick] (6,2.9) -- (6,2.3);
\draw [-latex,thick] (6,2.1) -- (6,1.5);
\draw [-latex,dotted] (5.7,1.1) -- (4.5,0.7);
\draw [-latex,dotted] (4.5,1.8) -- (5.7,1.3);

\draw (2,5) circle (0.3cm) node{\scriptsize $p_3$};
\draw [fill=black] (2.9,4.65) node[below]{\scriptsize $s_3$} rectangle (3.1,5.35);
\draw [-o,thick] (2,4.7) -- (2,2.1);
\draw [-latex,thick] (2.3,5) -- (2.9,5);
  \end{tikzpicture}
  
  \caption{Construction for deadlock-freeness}
  \label{fig:dlfpic}
\end{figure}
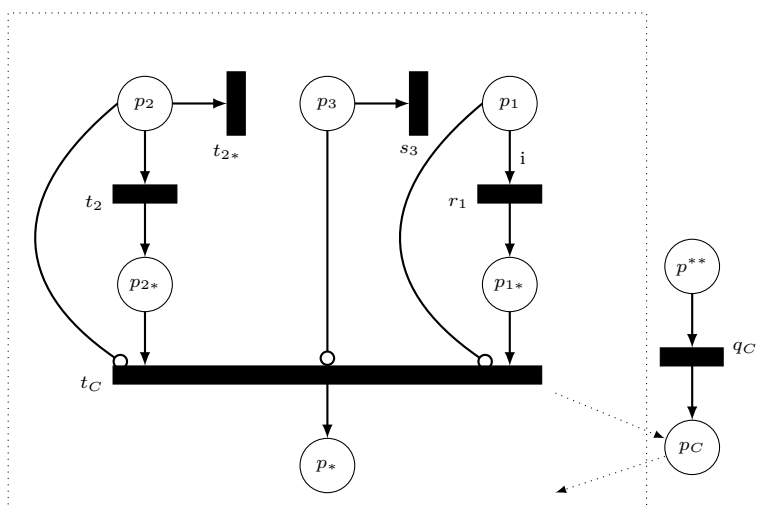

Given the original \Za{IR} net, for each conjunctive clause
considered, we perform the construction as shown in Figure~\ref{fig:dlfpic}. For every place $p_a\in A^C$, we add a construction as for
$p_2$ in Figure~\ref{fig:dlfpic}. For every place $p_i\in B^C_i$, we add a
construction as for $p_1$ in Figure~\ref{fig:dlfpic}. For all places $p\not\in
A^C\cup \bigcup_{i} B^C_i$, we add a construction as for $p_3$ in Figure~\ref{fig:dlfpic}. We call the transition $t_C$ in the figure as the "Check Transition", and refer to the set of transitions $r_i,s_i,t_i,t_{i*},t_C$ (excluding $q_C$) as \emph{transitions for   clause C}.  Note that for any $p_i\in P$, exactly one of
$r_i,s_i,t_i$ exist since the sets $A^C$ and the sets $B^C_i$ are all
pairwise disjoint.

Our construction also adds two new places, $p_C$ and $p^{**}$, and one
new transition $q_C$ such that
\begin{itemize}
	\item there is a pre-arc and a post-arc of weight 1 from
          $p^{**}$ to every transition in the original net.  Thus,
          transitions in original net can fire only $p^{**}$ has a
          token.
	\item there is a pre-arc of weight 1 from $p_c$ to every 
          transition for clause $C$ (within dotted box).
	\item there is a post-arc of weight 1 to $p_c$ from every
          transition for clause $C$ (within dotted box), except from
          $t_C$ to $p_C$.
\end{itemize}
Note that hierarchy is preserved in the transformed net, since the
only new transitions which have inhibitor/reset arcs are the
check transitions, which have inhibitor arcs from all places in the
original net.

Let $N$ be the original net in \Za{IR} with $P$ being its set of
places, and let $N'$ be the transformed net, also in \Za{IR}, obtained
above.  Define a mapping $f: \Markings{N} \rightarrow \Markings{N'}$
as follows: $f(M)(p)= M(p)$ if $p\in P$; $f(M)(p^{**}) = 1$ and
$f(M)(p) = 0$ in all other cases.  If $M_0$ is the initial marking in
$N$, define $M'_0 = f(M_0)$ to be the initial marking in $N'$.
\begin{claim}\label{claim:deadlock}
  The marking $M'_\star$ of $N'$, defined as $M'_\star(p) = 1$ if $p = p^*$ and
  $M'_\star(p) = 0$ otherwise, is reachable from $M'_0$ in $N'$ iff there exists a
  deadlocked marking reachable from $M_0$ in $N$.
\end{claim}
See Appendix \ref{deadlock-appendix} for proof of this claim.

\begin{theorem}
  \label{thm:hirpn-dlf}
Deadlock-freeness for \Za{IR} is decidable.
\end{theorem}
This follows from the above reduction and from the decidability of
reachability for \Za{IR}.  A detailed proof can be found in
\ref{deadlock-appendix}.

\section{Adding Reset Arcs without Hierarchy}
\label{sec:rhipn}
The previous section dealt with extension of Petri nets where reset arcs were added within the hierarchy of the inhibitor arcs. This section discusses the decidability results when we add reset arcs outside the hierarchy of inhibitor arcs. It turns out that termination remains decidable for this extension of Petri nets too.

\subsection{Termination in \Zb{R}{I}}\label{term}
Our main idea here is to use a modified finite reachability tree (FRT) construction to provide an algorithm for termination in \Zb{R}{I}. The usual FRT construction (see for instance \cite{finkel}) for Petri nets does not work for Petri nets with even a single (hence hierarchical) inhibitor arc. 

In this section, we provide a construction that tackles termination in Petri nets  with hierarchical inhibitor arcs, even in the presence of additional reset arcs:
\begin{theorem}
\label{thm:rhipn-term}
Termination is decidable for \Zb{R}{I}.
\end{theorem}
Consider a \Zb{R}{I} net $(P,T,F,\sqsubseteq,M_0)$. We start by introducing a few definitions.

\begin{definition}
For any place $p\in P$, we define the \emph{index of the place $p$ (Index($p$))} as the number of places $q\in P$ such that $q\sqsubseteq p$.
The definition of Index over places induces an Index among transitions too: 
For any transition $t\in T$, its \emph{index} is defined as $Index(t)=\max_{F(p,t)=I} Index(p)$ By convention, if there is no such place, then $Index(t)=0$.
\end{definition}
Given markings $M_1$ and $M_2$ and $i\in \mathbb{N}$, we say that $M_1$ and $M_2$ are \emph{$i$-Compatible (denoted $Compat_i(M_1,M_2)$)}  if $\forall p\in P\, Index(p)\leq i\implies M_1(p)=M_2(p)$.
\begin{definition}
Consider a run $M_2\ce{ ->[\rho]}M_1$. Let $t^*=\arg\!\max_{t\in \rho}Index(t)$. We define $Subsume(M_2,M_1,\rho)=M_2\leq M_1\wedge \bigg( Compat_{Index(t^*)}(M_1,M_2) \bigg)$
\end{definition}
To understand this definition note that if $\rho$ can be fired at $M_2$ and reaches $M_1$ and if $Subsume(M_2, M_1,\rho)$ is true, then at $M_1$, $\rho$ can be fired again. Note that in classical Petri nets without inhibitor arcs, $Subsume(M_2,M_1,\rho)=M_2\leq M_1$, and hence this is the classical monotonicity condition. But in the presence of even a single inhibitor arc, this may differ.

Given \Zb{R}{I} $N=(P,T,F,\sqsubseteq,M_0)$, we define the \emph{Extended Reachability Tree $\mathit{ERT}(N)$} as a directed unordered tree where the nodes are labelled by markings $M:P\rightarrow\N$, rooted at $n_0:M_0$ (initial marking). If $M_1\ce{ ->[t]}M_2$ for some markings $M_1$ and $M_2$ and transition $t\in T$, then a node marked by $n':M_2$ is a child of the node $n:M_1$. Consider any node labelled $M_1$. If along the path from root $n_0:M_0$ to $n:M_1$, there is a marking $n':M_2$ ($n\neq n'$), such that the path from $n':M_2$ to $n:M_1$ corresponds to run $\rho$ and $Subsume(M_2,M_1,\rho)$ is true, then $M_1$ is made a leaf node (which we call a \emph{subsumed} leaf node). Note that leaf nodes in this tree are of two types: either leaf nodes caused by subsumption as above or leaf nodes due to deadlock, where no transition is fireable. 
\begin{lemma}\label{2.0}
For any \Zb{R}{I} $N=(P,T,F,\sqsubseteq,M_0)$, $\mathit{ERT}(N)$ is finite.
\end{lemma}
\begin{proof}
Assume the contrary. By Konig's Lemma, there is an infinite path. Let the infinite path correspond to a run $\rho=M_0\Arrow{t_1}M_1\Arrow{t_2}M_2\dots \Arrow{t_i}M_i\dots$. 

Let $t\in T$ be the transition which has maximum index among the transitions which are fired infinitely often in run $\rho$. Thus all transitions having higher index than Index(t) fire only finitely many times. Let $b$ be chosen such that $\forall i\geq b\,Index(t_i)\leq Index(t)$ (i.e b is chosen after the last position where any transition with higher index than Index(t) fires). This exists by the definition of $t$. Since $t$ is fired infinitely often, the sequence $\{M_i|i>b\wedge t_{i+1}=t\}$ is an infinite sequence. As $\leq$ over markings is a well-quasi ordering, there exist two markings $M_i$ and $M_j$, such that both belong to the above sequence (i.e. $t_{i+1}=t_{j+1}=t$), $M_i\leq M_j$ and $i<j$. Now, since $t_{i+1}=t_{j+1}=t$, $$\forall p\in P,  Index(p)\leq Index(t)\implies M_i(p)=M_j(p)=0$$ for $t$ to fire at $M_i$ and $M_j$. Thus, $Compat_{Index(t)}(M_i,M_j)$ is true. Note that $t$ is the maximum index transition fired in the run from $M_i$ to $M_j$, since no higher index transition fires after position $b$ and $j>i>b$. Hence, $Subsume(M_i,M_j,\rho')$ is true, where $\rho'$ is the run from $M_i$ to $M_j$. But then, the path would end at $M_j$. Contradiction.
\end{proof}
Thus, we have shown that the $\mathit{ERT}$ is always finite. Next, we will show a crucial property of $Compat_i$, which will allow us to check for a non-terminating run.
\begin{lemma}\label{2.1}
Consider markings $M_1$ and $M_2$ such that $M_1\leq M_2$. Let $i\in \N$ be such that we have $Compat_i(M_1,M_2)$. Then for any run $\rho$ over the set of transitions $T_i=\{t|t\in T\wedge Index(t)\leq i\}$, if $M_1\Arrow{\rho}M'_1$, then $M_2\Arrow{\rho}M'_2$, where $M'_1\leq M'_2$ and $Compat_i(M'_1,M'_2)$.
\end{lemma}
\begin{proof}
We can prove this by induction. We first prove that $t$ is firable at $M_2$. If $F(p,t)\in \N$, then $M_2(p)\geq M_1(p)\geq F(p,t)$. If $F(p,t)=I$, i.e., it is an inhibitor arc, then $Index(p)\leq Index(t)\leq i$. But now, since $Compat_i(M_1,M_2)$ holds and $t$ is firable at $M_1$, we obtain $M_2(p)=M_1(p)=0$. Finally, if $F(p,t)=R$, i.e., it is a reset arc, then there is no condition on $M_2(p)$ for $t$ to be firable. Hence, $t$ is firable at $M_2$.

Now let $M_2\Arrow{t}M'_2$. Then, for all $p\in P$, $M'_2(p)=M_2(p)-F(p,t)+F(t,p)$ and $M'_1(p)=M_1(p)-F(p,t)+F(t,p)$. Since $F(t,p)$ is constant and $F(p,t)$ can depend only on number of tokens in place $p$ (so, if $M_1(p)$ and $M_2(p)$ were equal before firing, they remain equal now), we obtain that $Compat_i(M'_1,M'_2)$ and $M'_1\leq M'_2$.
\end{proof}
\begin{lemma}\label{2.2}
\Zb{R}{I} $N$ has a non-terminating run iff $\mathit{ERT}(N)$ has a subsumed leaf node.
\end{lemma}
\begin{proof}($\implies$) If \Zb{R}{I} $N$ has a non-terminating run, then $\mathit{ERT}(N)$ has a subsumed leaf node. Consider a non-terminating run. This run has a finite prefix in $\mathit{ERT}(N)$. This prefix ends in a leaf that is not a deadlock (as some transition is firable). Thus it is a subsumed leaf node.
  
\noindent($\Longleftarrow$) If $\mathit{ERT}(N)$ has a subsumed leaf node, then $N$ has a non-terminating run. To see this, consider any subsumed leaf node labelled by marking $M_2$. Let $M_1$ be the marking along the path $M_0$ to $M_2$, and $\rho$ be the run from $M_1$ to $M_2$, such that $Subsume(M_2,M_1,\rho)$ is true. Hence, we have $M_1\Arrow{\rho}M_2$. Take $t^*=argmax_{t\in \rho}Index(t)$ and $i=Index(t^*)$. Since $Subsume(M_1,M_2,\rho)$ is true, we have $M_1\leq M_2$ and $Compat_i(M_1,M_2)$ is true. We also have $\rho$ is a run over $T_i=\{t|t\in T\wedge Index(t)\leq i\}$(by definition of $i$).\\
Thus, by Lemma \ref{2.1}, we have $M_2\Arrow{\rho}M_3$, where $M_2\leq M_3$ and $Compat_i(M_2,M_3)$ is true. Thus, $\rho$ can again be fired at $M_3$ and so on, resulting in a non-terminating run.
\end{proof}
Finally, we observe that checking $Subsume(M_2,M_1,\rho)$ is also easily doable. Thus, for any \Zb{R}{I} net, one can draw its extended reachability tree and decide the termination problem using the ERT. This completes the proof of the theorem. We observe here that this construction cannot be immediately lifted to checking boundedness due to the presence of reset arcs. However, we can lift this to check for termination in HIRPN and R-HIRPN.
 

\subsection{Coverability in \Zb{R}{I}}\label{Cover}
While termination turned out to be decidable, reachability is undecidable for \Zb{R}{I} nets in general (since it subsumes reset Petri  nets). Indeed \cite{transfer} show that reachability is undecidable for Petri nets with 2 reset arcs. Using a similar strategy, in Appendix~\ref{1I1RPN}, we tighten the undecidability result to show that reachability in Petri nets with one inhibitor arc and one reset arc is undecidable. Further, we can modify the construction presented, to show that Deadlockfreeness in Petri nets with one reset arc and one inhibitor arc is undecidable too. 

Next we turn our attention to coverability problem and show the following result.
\begin{theorem}\label{cover-thm}
Coverability is undecidable for Petri nets with two reset/transfer arcs and an inhibitor arc.
\end{theorem}
The rest of this section proves the above theorem. To do this, we construct a Petri net with two reset arcs, one inhibitor arc that simulates the two counter Minsky Machine. The Minsky Machine $M$ is defined as follows - It has a finite set of instructions $q_{i}, 0 \le i \le n$, $q_{0}$ is the initial state and $q_{n}$ is the final instruction i.e. there are no transition rules from $q_{n}$. There are two counters $C_{1}$ and $C_{2}$ in the machine. There are two kind of transitions.  
\begin{enumerate}
\item 
INC(r,j) - Increase $C_{r}$, by 1, Go to $q_{j}$. r can be 1 or 2.  
\item 
JZDEC(r,j,l) - If $C_{r}$ is zero, Go to $q_{l}$, else decrease $C_{r}$ by 1 and go to $q_{j}$. r can be 1 or 2. 
It is well known that reachability of $q_{n}$ in a Minsky Machine is undecidable. 
\end{enumerate}

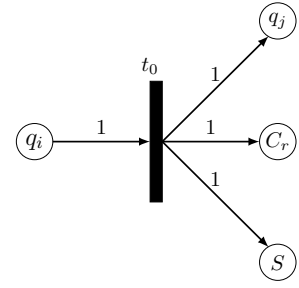
\begin{wrapfigure}[12]{r}{0.3\textwidth}
\scalebox{0.8}{
\begin{tikzpicture}
\draw (0,2) circle (0.3cm) node{\Large $q_i$};
\draw (4,2) circle (0.3cm) node{$C_r$};
\draw (4,0) circle (0.3cm) node{$S$};
\draw (4,4) circle (0.3cm) node{$q_j$};
\draw [fill=black] (1.9,3) node[above]{$t_{0}$} rectangle (2.1,1);
\draw[-latex,thick] (0.3,2) -- node[above]{1} (1.9,2);
\draw[-latex,thick] (2.1,2) -- node[above]{1} (3.85,3.75);
\draw[-latex,thick] (2.1,2) -- node[above]{1} (3.7,2);
\draw[-latex,thick] (2.1,2) -- node[above]{1} (3.85,0.25);
\end{tikzpicture}
}
\caption{Increment}
\label{fig:incr}
\end{wrapfigure}
We encode $M$ into Petri net $P$ as follows - we use places $q_{i}, 0 \le i \le n$ to encode each instruction. The place $q_{i}$ gets a token when we simulate instruction $i$ in the Minsky Machine. We use two places $C_{1}$ and $C_{2}$ to store the number of tokens corresponding to the counter values in $C_{1}$ and $C_{2}$ in the counter machine. We use special place $S$ which stores the sum of $C_{1}$ and $C_{2}$. The adjoining Figure~\ref{fig:incr} shows the construction for increment. When $q_{i}$ gets a token, the transition is fired, $C_{r}$ and $S$ are incremented by 1 and $q_{j}$ gets the token to proceed. 

Next, to simulate decrement (along with zero check), i.e., if $C_r=0$, then $q_\ell$, else $q_j$), we introduce non-determinism in the Petri net.  The gadget for this is shown in figure \ref{fig:decrement}. 

When we reach a decrement with zero check instruction, we guess whether $C_r$ is zero, and if so, fire $t_{11}$ and then $t_3$. Else we decrement it $C_r$ by 1 and fire $t_2$. We have two cases:
\begin{itemize}
 \item 
\textbf{Case - 1} : If $C_{r}$ is actually zero, it runs correctly as $t_2$ would not fire. The transition $t_3$ fires and $C_{r}$ remains zero. And $q_{l}$ gets the token.  
\item
\textbf{Case - 2} : If $C_{r}$ has non-zero tokens, both transitions can fire. But the runs in which $t_3$ fires are ``wrong'' runs. We call such transitions as \emph{Incorrect} transitions. The crucial point is that in runs with incorrect transitions, $S$ is not decremented where as $C_r$ is decremented. Hence $M(S)\neq M(C_1)+M(C_2)$ in markings reached by runs with incorrect transitions. 
\end{itemize} 

Note that in any run of $P$, $q_{i}$ and only $q_{i}$ in the Petri  net gets a token when the instruction numbered $i$ is being simulated. Now, we have the following lemma which proves the correctness of the reduction.

\begin{lemma}
In any run of $P$ reaching marking $M$, $M(S) \geq M(C_{1})+M(C_{2})$ and $M(S) = M(C_{1})+M(C_{2})$ iff there are no incorrect transitions.
\end{lemma}

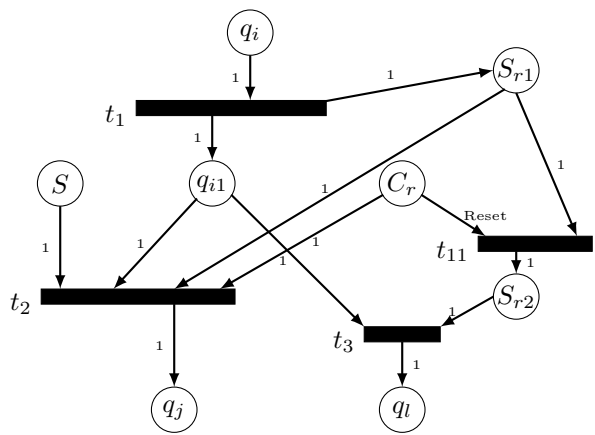
\begin{figure}[t]
\begin{center}
  \begin{tikzpicture}
\draw (4,6) circle (0.3cm) node{$q_i$};
\draw (1.5,4) circle (0.3cm) node{$S$};
\draw (3.5,4) circle (0.3cm) node{$q_{i1}$};
\draw (6,4) circle (0.3cm) node{$C_r$};
\draw (7.5,5.5) circle (0.3cm) node{$S_{r1}$};
\draw (3,1) circle (0.3cm) node{$q_j$};
\draw (6,1) circle (0.3cm) node{$q_l$};
\draw (7.5,2.5) circle (0.3cm) node{$S_{r2}$};

 \draw [fill=black] (2.5,4.9) node[left]{$t_1$} rectangle (5,5.1);
 \draw [fill=black] (1.25,2.4) node[left]{$t_2$} rectangle (3.8,2.6);
\draw [fill=black] (5.5,1.9) node[left]{$t_3$} rectangle (6.5,2.1);
\draw [fill=black] (7,3.1) node[left]{$t_{11}$} rectangle (8.5,3.3);
 \draw[-latex,thick] (4.5,5) -- node[above]{\tiny 1} (7.2,5.5);
 \draw[-latex,thick] (1.5,3.7) -- node[left]{\tiny 1} (1.5,2.6);
 \draw[-latex,thick] (3.3,3.8) -- node[left]{\tiny 1} (2.2,2.6);
 \draw[-latex,thick] (3,2.4) -- node[left]{\tiny 1} (3,1.3);
\draw[-latex,thick] (6.25,3.85) -- node[right]{\tiny Reset} (7.1,3.3);
\draw[-latex,thick] (7.5,5.2) -- node[right]{\tiny 1} (8.3,3.3);
\draw[-latex,thick] (7.5,3.1) -- node[right]{\tiny 1} (7.5,2.8);

\draw[-latex,thick] (7.2,2.5) -- node[left]{\tiny 1} (6.5,2.1);
\draw[-latex,thick] (7.35,5.25) -- node[left]{\tiny 1} (3.0,2.6);

\draw[-latex,thick] (3.75,3.85) -- node[left]{\tiny 1} (5.5,2.1);
\draw[-latex,thick] (4,5.7) -- node[left]{\tiny 1} (4,5.1);
\draw[-latex,thick] (3.5,4.9) -- node[left]{\tiny 1} (3.5,4.3);
\draw[-latex,thick] (6,1.9) -- node[right]{\tiny 1} (6,1.3);

\draw[-latex,thick] (5.75,3.85) -- node[right]{\tiny 1} (3.6,2.6);
\end{tikzpicture}
\caption{Decrement along with zero check}
\label{fig:decrement}
\end{center}
\end{figure}

If the Minsky Machine reaches instruction $q_{n}$, we reach the place $q_{n}$ state in Petri net. But, if the Minsky Machine doesn't reach $q_{n}$, there is a chance that we reach $q_{n}$ in Petri net because of incorrect transitions. By the above lemma, to check if there had been any incorrect transitions along the run, we just check at the end (at $q_n$) if $M(S)=M(C_1)+M(C_2)$, which we can do using an inhibitor arc.  
Thus $q_{n+1}$ gets tokens iff the Minsky Machine reaches the instruction $q_{n}$. Hence reaching instruction $q_n$ in Minsky Machine is equivalent to asking if we can cover the marking in which all places except $q_{n+1}$ have 0 tokens and $q_{n+1}$ has 1 token. We also note that the above proof holds for undecidability of coverability in Petri nets with 2 transfer arcs and an inhibitor arc. The proof of the above lemmas, the inhibitor arc construction and extension to transfer arcs are presented in Appendix \ref{2R1C}.
Finally, the problem of coverability in Petri nets with 1 inhibitor arc and 1 reset arc is open.  

\begin{problem}
Is coverability in Petri nets with 1 reset arc and 1 inhibitor arc decidable ?
\end{problem}

\section{Adding Transfer Arcs within and without Hierarchy}
\label{sec:transf}
\subsection{Reachability and deadlock-freeness in \Za{T}}\label{transfer_undec}
We show a reduction from Petri nets with 2 (non-hierarchical) transfer arcs to \Za{T} preserving reachability and deadlock-freeness. Since reachability and deadlock-freeness in Petri nets with 2 transfer arcs are undecidable~\cite{transfer}, they are undecidable in \Za{T} too.
\begin{theorem}
\label{thm:htpn-reach-dlf}
Reachability and deadlock-freeness are undecidable in \Za{T}.
\end{theorem}
\begin{proof}
Given a Petri net with 2 transfer arcs, we will construct a \Za{T} such that reachability of markings and deadlock-freeness is preserved.  Let $N$ be such a net as shown in the figure below, with two transfer arcs, one from $p_1$ to $p_3$ via $t_1$ and another from $p_2$ to $p_4$ via $t_2$. $t_3$ and $t_4$ are representative of any other transitions to and from $p_1$. Wlog., we assume that there is no arc from $p_1$ to $t_2$. If this is not the case, we can add a place and transition in between to create an equivalent net (while adding no deadlocked reachable marking), see Appendix \ref{transfer_undec-appendix}. Now, the construction is shown in the diagram below.

\begin{changemargin}{-1.5cm}{-1.5cm}
\begin{minipage}{.5\textwidth}
\begin{tikzpicture}[scale=0.7]
\draw (3,18.5) circle (0.3cm) node{$p_1$};
\draw (1,15.5) circle (0.3cm) node{$p_3$};
\draw [fill=black] (0.65,16.9) node[left]{$t_1$} rectangle (1.35,17.1);
\draw [fill=black] (2.65,16.9) node[left]{$t_3$} rectangle (3.35,17.1);
\draw [fill=black] (4.65,16.9) node[left]{$t_4$} rectangle (5.35,17.1);
\draw[-latex,thick] (3,18.2) .. controls (1,17.7) .. node[right]{tf} (1,15.8);
\draw[-latex,thick] (3,17.1) -- node[right]{} (3,18.2);
\draw[-latex,thick] (3,18.2) -- node[right]{} (5,17.1);
\draw (8,15.5) circle (0.3cm) node{$p_2$};
\draw (7,15.5) circle (0.3cm) node{$p_4$};
\draw [fill=black] (6.65,16.9) node[left]{$t_2$} rectangle (7.35,17.1);
\draw[-latex,thick] (8,15.8) .. controls (8,18.7) and (7,18.7) .. node[right]{tf} (7,15.8);
\end{tikzpicture}
\end{minipage}
\begin{minipage}{.5\textwidth}
\begin{tikzpicture}[scale=0.7]
\draw (-0.5,15) -- (-0.5,8);
\draw (3,14.5) circle (0.3cm) node{$p_1$};
\draw (1,11.5) circle (0.3cm) node{$p_3$};
\draw [fill=black] (0.65,12.9) node[left]{$t_1$} rectangle (1.35,13.1);
\draw [fill=black] (2.65,12.9) node[left]{$t_3$} rectangle (3.35,13.1);
\draw [fill=black] (4.65,12.9) node[left]{$t_4$} rectangle (5.35,13.1);
\draw[-latex,thick] (3,14.2) .. controls (1,13.7) .. node[right]{tf} (1,11.8);
\draw[-latex,thick] (3,13.1) -- node[right]{} (3,14.2);
\draw[-latex,thick] (3,14.2) -- node[right]{} (5,13.1);
\draw (10,11.5) circle (0.3cm) node{$p_2$};
\draw (9,11.5) circle (0.3cm) node{$p_4$};
\draw (7,14.5) circle (0.3cm) node{$p_*$};
\draw [fill=black] (8.65,12.9) node[left]{$t_2$} rectangle (9.35,13.1);
\draw[-latex,thick] (10,11.8) .. controls (10,14.7) and (9,14.7) .. node[right]{tf} (9,11.8);
\draw[dotted] (0,12.5) rectangle (6,13.5);
\draw[-latex,dotted] (6.7,14.5) .. controls (5.5,14.5) .. (5.5,13.2);
\draw[-latex,dotted] (5.7,13) .. controls (7,13) .. (7,14.2);
\draw[-latex,thick] (7.3,14.5) .. controls (8.8,14.5) .. node[right]{} (8.8,13.1);
\draw[-latex,thick] (8.8,12.9) -- node[left]{} (7,8.8);
\draw (3,8.5) circle (0.3cm) node{$p'_1$};
\draw [fill=black] (0.65,10.1) node[left]{$t'_1$} rectangle (1.35,9.9);
\draw [fill=black] (2.65,10.1) node[left]{$t'_3$} rectangle (3.35,9.9);
\draw [fill=black] (4.65,10.1) node[left]{$t'_4$} rectangle (5.35,9.9);
\draw[-latex,thick] (3,8.8) .. controls (1,9.3) .. node[right]{tf} (1,11.2);
\draw[-latex,thick] (3,9.9) -- node[right]{} (3,8.8);
\draw[-latex,thick] (3,8.8) -- node[right]{} (5,9.9);
\draw (7,8.5) circle (0.3cm) node{$p'_*$};
\draw [fill=black] (8.65,10.1) node[left]{$t'_2$} rectangle (9.35,9.9);
\draw[-latex,thick] (10,11.2) .. controls (10,8.3) and (9,8.3) .. node[right]{tf} (9,11.2);
\draw[dotted] (0,9.5) rectangle (6,10.5);
\draw[-latex,dotted] (6.7,8.5) .. controls (5.5,8.5) .. (5.5,9.8);
\draw[-latex,dotted] (5.7,10) .. controls (7,10) .. (7,8.8);
\draw[-latex,thick] (7.3,8.5) .. controls (8.8,8.5) .. node[right]{} (8.8,9.9);
\draw[-latex,thick] (8.8,10.1) -- node[left]{} (7,14.2);
\end{tikzpicture}
\end{minipage}
\end{changemargin}\\
Six transfer arcs have not been shown in the construction above. These are the following:\\
\begin{minipage}{0.5\textwidth}
\begin{itemize}
    \item From $p_1$ to $p_3$ through $t'_1$.
    \item From $p_1$ to $p'_1$ through $t_2$.
    \item From $p_1$ to $p'_1$ through $t'_2$.
\end{itemize}
\end{minipage}
\begin{minipage}{0.5\textwidth}
\begin{itemize}
    \item From $p'_1$ to $p_3$ through $t_1$.
    \item From $p'_1$ to $p_1$ through $t_2$.
    \item From $p'_1$ to $p_1$ through $t'_2$.
\end{itemize}
\end{minipage}
These transfer arcs ensure hierarchy among the transfer arcs with the ordering $p_1<p'_1<p_2$. The dotted arc from $p_*$ to the upper dotted box represents a pre-arc from $p_*$ to every transition in the box. Similarly, we have an arc from every transition in the box to $p_*$. Similarly, we have arcs for the lower dotted box and $p'_*$ also. The intuitive idea behind the construction is to represent the place $p_1$ in the original net by two places $p_1$ and $p'_1$ in the modified net. At every marking, $p_1$ of original net is represented by one of the two places $p_1$ or $p'_1$ in the modified net. $p'_*$ and $p_*$ are used to keep track of which place represents $p_1$ in current marking. Everytime transition $t_2$ fires, the representative place swaps. Let the original net be $(P,T,F)$ and the constructed net be $(P',T',F')$. The initial marking $M'_0$ is given by $M'_0(p_*)=1$, $M'_0(p'_*)=M'_0(p'_1)=0$, and $M'_0(p)=M_0(p)$ for all other $p\in P$. 
Now, given marking $M$ of original net, let us define the set $S^{ext}=\{A_M,B_M\}$, where,\\
\begin{minipage}{0.5\textwidth} 
$$A_M(p)=\begin{cases} M(p)\,\,\,\,\,p \not \in \{p_1,p'_1,p'_*,p_*\} \\1\,\,\,\,\,\,\,\,\,\,\,\,\,\,\,\,p=p_*\\0\,\,\,\,\,\,\,\,\,\,\,\,\,\,\,\,p\in\{p'_*,p'_1\}\\M(p_1)\,\,\,\,\,p=p_1\end{cases}$$
\end{minipage}
\begin{minipage}{0.5\textwidth}
$$B_M(p)=\begin{cases} M(p)\,\,\,\,\,p \not \in \{p_1,p'_1,p'_*,p_*\} \\1\,\,\,\,\,\,\,\,\,\,\,\,\,\,\,\,p=p'_*\\0\,\,\,\,\,\,\,\,\,\,\,\,\,\,\,\,p\in\{p_*,p_1\}\\M(p_1)\,\,\,\,\,p=p'_1\end{cases}$$
\end{minipage}
\begin{claim}
Marking $A_M$ or $B_M$ is reachable from $M'_0$ in the constructed net iff marking $M$ is reachable from $M_0$ in the original net. 
\end{claim}
From this claim (proof in Appendix \ref{transfer_undec-appendix}), we obtain the proof of the theorem.
\end{proof}
\begin{corollary}\label{hitpncover}
Coverability is undecidable in \Za{IT}.
\end{corollary}
\begin{proof}
From Theorem \ref{cover-thm}, coverability is undecidable in Petri nets with two transfer arcs and one inhibitor arc. Given such a net $N$, we can perform a construction similar to above, to reduce the coverability problem in $N$ to coverability problem in a \Za{IT} net.
\end{proof}

\subsection{Hardness of Termination in \Za{IT}}
\label{sec:skolem}
Termination in \Za{IcT} is decidable as shown in Section \ref{term}. Termination in \Za{T} is also decidable, as it is known that termination in transfer Petri nets is decidable. However, it turns out that termination in \Za{IT} which subsumes the above two problems is as hard as the positivity problem which is a long standing open problem about linear recurrent sequences (\cite{positivity2},\cite{positivity1}). In the following we show the reduction from the positivity problem to the problem of termination in \Za{IT}.

\begin{definition}[Positivity Problem]
Given a matrix $M \in \mathbb{Z}^{n \times n}$ and a vector $v_0 \in \mathbb{Z}^n$,
is $M^k v_0 \geq 0$ for all $k \in \mathbb{N}$? 
\end{definition}

Given matrix $M \in \mathbb{Z}^{n \times n}$ and vector $v_0 \in \mathbb{Z}^n$,
we construct a net $N \in \Za{IT}$ such that $N$ does not terminate iff $M^k v_0 \geq 0$
for all $k \in \mathbb{N}$.
Consider the following
$\mathsf{while}$ loop program
\lstinline{v = v0; while (v >= 0) v = Mv}. Clearly, this program is non-terminating
iff $M^k v_0 \geq 0$ for all $k$. We construct a net $N$ which simulates
this linear program. $N$ contains two phases, a forward phase that has
the effect of multiplying $v$ by $M$, and a backward phase that takes the role
of assignment, i.e. assigning the new vector $M v$ computed in the forward phase
back to $v$. We also check for non-negativity in the backward phase, and design
the net $N$ to terminate if any component goes negative.
\begin{figure}[t]
\begin{minipage}{0.5\textwidth}
\centering
\scalebox{0.75}{
\begin{tikzpicture}
\draw (1,6) circle (0.3cm) node{$u_1$};
\draw (1,3) circle (0.3cm) node{$u_2$};
\draw (1,0) circle (0.3cm) node{$u_3$};
\draw [fill=black] (2,6.35) node[above]{$t_1$} rectangle (2.2,5.65);
\draw [fill=black] (2,3.35) node[above]{$t_2$} rectangle (2.2,2.65);
\draw [fill=black] (2,0.35) node[above]{$t_3$} rectangle (2.2,-0.35);
\draw (4,7) circle (0.3cm) node{$u_{11}$};
\draw (4,6) circle (0.3cm) node{$u_{12}$};
\draw (4,5) circle (0.3cm) node{$u_{13}$};
\draw (4,4) circle (0.3cm) node{$u_{21}$};
\draw (4,3) circle (0.3cm) node{$u_{22}$};
\draw (4,2) circle (0.3cm) node{$u_{23}$};
\draw (4,1) circle (0.3cm) node{$u_{31}$};
\draw (4,0) circle (0.3cm) node{$u_{32}$};
\draw (4,-1) circle (0.3cm) node{$u_{33}$};
\draw [fill=black] (7,7.25) node[above]{$t_{11}$} rectangle (7.2,6.75);
\draw [fill=black] (7,6.25) node[above]{$t_{12}$} rectangle (7.2,5.75);
\draw [fill=black] (7,5.25) node[above]{$t_{13}$} rectangle (7.2,4.75);
\draw [fill=black] (7,4.25) node[above]{$t_{21}$} rectangle (7.2,3.75);
\draw [fill=black] (7,3.25) node[above]{$t_{22}$} rectangle (7.2,2.75);
\draw [fill=black] (7,2.25) node[above]{$t_{23}$} rectangle (7.2,1.75);
\draw [fill=black] (7,1.25) node[above]{$t_{31}$} rectangle (7.2,0.75);
\draw [fill=black] (7,0.25) node[above]{$t_{32}$} rectangle (7.2,-0.25);
\draw [fill=black] (7,-0.75) node[above]{$t_{33}$} rectangle (7.2,-1.25);
\draw (11,6) circle (0.3cm) node{$u'_1$};
\draw (11,3) circle (0.3cm) node{$u'_2$};
\draw (11,0) circle (0.3cm) node{$u'_3$};
\draw (11,9) circle (0.3cm) node{$G'$};
\draw (4,9) circle (0.3cm) node{$G$};
\draw[-latex,thick] (1.3,3) -- node[above]{1} (2,3);
\draw[-latex,thick] (1.3,6) -- node[above]{1} (2,6);
\draw[-latex,thick] (1.3,0) -- node[above]{1} (2,0);

\draw[-latex,thick] (2.2,3) -- node[above]{4} (3.7,4);
\draw[-latex,thick] (2.2,3) -- node[above]{5} (3.7,3);
\draw[-latex,thick] (2.2,3) -- node[above]{6} (3.7,2);
\draw[-latex,thick] (2.2,6) -- node[above]{1} (3.7,7);
\draw[-latex,thick] (2.2,6) -- node[above]{2} (3.7,6);
\draw[-latex,thick] (2.2,6) -- node[above]{3} (3.7,5);
\draw[-latex,thick] (2.2,0) -- node[above]{7} (3.7,1);
\draw[-latex,thick] (2.2,0) -- node[above]{8} (3.7,0);
\draw[-latex,thick] (2.2,0) -- node[above]{9} (3.7,-1);
\draw[-latex,thick] (4.3,-1) -- node[above]{1} (7,-1);
\draw[-latex,thick] (4.3,0) -- node[above]{1} (7,0);
\draw[-latex,thick] (4.3,1) -- node[above]{1} (7,1);
\draw[-latex,thick] (4.3,2) -- node[above]{1} (7,2);
\draw[-latex,thick] (4.3,3) -- node[above]{1} (7,3);
\draw[-latex,thick] (4.3,4) -- node[above]{1} (7,4);
\draw[-latex,thick] (4.3,5) -- node[above]{1} (7,5);
\draw[-latex,thick] (4.3,6) -- node[above]{1} (7,6);
\draw[-latex,thick] (4.3,7) -- node[above]{1} (7,7);
\draw[-latex,thick] (7.2,7) node[xshift=0.5cm,yshift=0.1cm]{1} -- (11,6.3);
\draw[-latex,thick] (7.2,6) node[xshift=0.5cm,yshift=-0.2cm]{1} -- (11,3.3);
\draw[-latex,thick] (11,0.3) node[xshift=-0.8cm,yshift=0.6cm]{1} -- (7.2,5);
\draw[-latex,thick] (10.7,6) node[xshift=-0.5cm,yshift=0cm]{1} -- (7.2,4);
\draw[-latex,thick] (10.7,3) node[xshift=-0.5cm,yshift=0.2cm]{1} -- (7.2,3);
\draw[-latex,thick] (10.7,0) node[xshift=-0.5cm,yshift=0cm]{1} -- (7.2,2);
\draw[-latex,thick] (7.2,1) node[xshift=0.5cm,yshift=0.3cm]{1} -- (11,5.7);
\draw[-latex,thick] (11,2.7) node[xshift=-0.8cm,yshift=-0.4cm]{1} -- (7.2,0);
\draw[-latex,thick] (7.2,-1) node[xshift=0.5cm,yshift=-0.1cm]{1} -- (11,-0.3);
\draw[-latex,dotted] (7.2,7) -- node[above]{6} (11,8.7);
\draw[-latex,dotted] (7.2,6) -- node[above]{15} (11,8.7);
\draw[-latex,dotted] (11,8.7) -- node[above]{24} (7.2,5);
\draw[-latex,dotted] (11,8.7) -- node[above]{6} (7.2,4);
\draw[-latex,dotted] (11,8.7) -- node[above]{15} (7.2,3);
\draw[-latex,dotted] (11,8.7) -- node[above]{24} (7.2,2);
\draw[-latex,dotted] (7.2,1) -- node[above]{6} (11,8.7);
\draw[-latex,dotted] (11,8.7) -- node[above]{15} (7.2,0);
\draw[-latex,dotted] (7.2,-1) -- node[above]{24} (11,8.7);
\draw[-latex,dotted] (4,8.7) -- node[above]{1} (7,0);
\draw[-latex,dotted] (4,8.7) -- node[above]{1} (7,1);
\draw[-latex,dotted] (4,8.7) -- node[above]{1} (7,2);
\draw[-latex,dotted] (4,8.7) -- node[above]{1} (7,3);
\draw[-latex,dotted] (4,8.7) -- node[above]{1} (7,4);
\draw[-latex,dotted] (4,8.7) -- node[above]{1} (7,5);
\draw[-latex,dotted] (4,8.7) -- node[above]{1} (7,6);
\draw[-latex,dotted] (4,8.7) -- node[above]{1} (7,7);
\draw[-latex,dotted] (4,8.7) -- node[above]{1} (7,-1);
\draw[dotted] (0.5,-2) -- node[above]{$G_1$} (2.5,-2) -- (2.5,8) -- (0.5,8) -- (0.5,-2);
\draw[dotted] (3.5,-2) -- node[above]{$G_2$} (7.5,-2) -- (7.5,8) -- (3.5,8) -- (3.5,-2);
\draw[dotted] (10.5,-2) -- node[above]{$G_3$} (11.5,-2) -- (11.5,8) -- (10.5,8) -- (10.5,-2);
\end{tikzpicture}
}
\caption{Forward phase}
\label{fig:forward_phase}
\end{minipage}
\hspace{1.5cm}
\begin{minipage}{0.5\textwidth}
\centering
\scalebox{0.75}{
\begin{tikzpicture}

\tikzstyle{place}=[draw=none]

\node[draw,circle] at  (6,4) (A) {$u_1$};
\node[draw,circle] at (6,2) (B) {$u_2$};
\node[draw,circle] at (6,0) (C) {$u_3$};

\node[draw,circle] at  (3,4) (A1) {$u_1'$};
\node[draw,circle] at (3,2) (B1) {$u_2'$};
\node[draw,circle] at (3,0) (C1) {$u_3'$};

\node[draw,circle] at (3,6) (G1) {$G'$};
\node[draw,circle] at (6,6) (G) {$G$};

\draw [fill=black] (4.4,1) node[below]{$t_R$} rectangle (4.6,5);

\draw[->,thick] (A1) -- (A) node[near start, above] {};
\draw[->,thick] (B1) -- (B) node[near start, above] {};
\draw[->,thick] (C1) to[out=80,in=100] (C) node[above] {};
\draw[->,thick] (G1) to[out=-80,in=-100] (G) node[above] {};

\draw[-o,thick] (G) to[out=-160,in=130] node[above]{} (4.4,4.5);

\end{tikzpicture}
}
\caption{Backward phase: Arc from $G$ to $t_R$ is an inhibitor arc, rest are transfer arcs}
\label{fig:backward_phase}
\end{minipage}

\end{figure}

\medskip
\noindent\textbf{Forward Phase:} The construction of the forward phase petri net for a general matrix is explained below. An example of the construction is shown in Figure \ref{fig:forward_phase} for the matrix: 
\[
M=
  \begin{bmatrix}
    1 & -4 & 7\\
    2 & -5 & -8\\
    -3 & -6 & 9
  \end{bmatrix}
\]

We have n places, $u_1,u_2, \ldots,u_n$, corresponding to the n components of vector $v$. 
Each place $u_i$ is connected to a transition $t_i$ with a pre-arc weight of $1$.
Each $t_i$ also has a post-arc to a new place $u_{ij}$ for $1 \leq i, j \leq n$ with a weight $|M_{ji}|$,
i.e. the absolute value of the $(j, i)^{th}$ entry of matrix $M$, corresponding to $v_i$'s
contribution to the new value of $v_j$.
Finally, we have places $u'_1,u'_2, \ldots,u'_n$, corresponding to the n components
of the new value of vector $v$. Each place $u_j'$ is connected to place $u_{ij}$ by a
transition $t_{ij}$, with both the arcs weighted 1. If $M_{ji} \geq 0$, then $u_{ij}$ has a pre-arc
to $t_{ij}$ and $t_{ij}$ has a post-arc to $u_j'$. This has the effect of adding the value
of $u_{ij}$ to $u_j'$. On the other hand, if $M_{ji} < 0$, then
both $u_{ij}$ and $u_j'$ have pre-arcs to $t_{ij}$, which has the effect of subtracting
value of $u_{ij}$ from $u_j'$.

This simulates the forward phase, in effect multiplying the vector $v$, represented
by $u_i$'s in Figure~\ref{fig:forward_phase} by $M$ and storing the new components
in $u_i'$'s. To simulate the while loop program, we need to copy back each $u_i'$
to $u_i$, while performing the check that each $u_i'$ is non-negative.

\textbf{Backward phase :}
The copy back in backward phase (Fig. \ref{fig:backward_phase}) is demonstrated by a transfer arc from $u_i'$ to $u_i$ via
transition $t_R$. To ensure that the backward phase starts only after
the forward phase completes, (else, partially computed values would be copied back), we introduce a new place
$G$. $G$ stores as many tokens as the total number of times each transition
$t_{ij}$ will fire and has a pre-arc weighted 1 to each transition $t_{ij}$.
The emptiness of $G$ ensures that each $t_{ij}$ has
completed its firings in the current loop iteration. An inhibitor arc from $G$
to $t_R$ ensures that the forward phase completes before $t_R$ fires. We introduce a place $G'$ which computes the initial value of $G$ for next loop.
$G'$ has an arc connected
to $t_{ij}$ with weight $\sum_{k=1}^n |M_{kj}|$. If $u_j'$ has a pre-arc to $t_{ij}$,
then $G'$ has a pre-arc to $t_{ij}$ while if $t_{ij}$ has a post-arc to $u_j'$, then
it also has a post-arc to $G'$.
Finally, there is a transfer arc from $G'$ to $G$ via $t_R$.
Once
the forward phase finishes, the place $G$ is empty, hence, the only transition that
can fire is $t_R$, which completes the backward phase in one firing.
Combining the forward and backward phases, we obtain a net $N$ which simulates
the while loop program. The initial
marking assigns $(v_0)_i$, i.e. the $i$-th component of vector $v_0$ to place $u_i$,
and $\sum_{1 \leq i\leq n} (\sum_{1\leq j\leq n}|M_{ji}|)(v_0)_{i}$ tokens to $G$, while all other places are
assigned 0 tokens. The below lemma (see appendix \ref{appendix-62}) relates termination of $N$ with the
Positivity problem. 

\begin{lemma}
\label{nonterminating}
There exists a non-terminating run in $N$ iff $M^k v_0 \geq 0$ for all $k \in \mathbb{N}$. 
\end{lemma}
From the above lemma, we derive the following theorem. Note that as we have only one transition with inhibitor and transfer arcs, $N$ is a \Zb{T}{I} as well as \Za{IT}. 
\begin{theorem}
\label{thm:skolem-hard}
Termination in \Za{IT} as well as \Zb{T}{I} is as hard as the positivity problem.
\end{theorem}

\section{Conclusion}
In this paper, we investigated the effect of hierarchy on Petri nets extended with not only inhibitor arcs (as classically considered), but also reset and transfer arcs. For four of the standard decision problems, we settled the decidability for almost all these extensions using different reductions and proof techniques. As future work, we are interested in questions of boundedness and place-boundedness in these extended classes. We would also like to explore further links to problems on linear recurrences. We leave open one technical question of coverability for Petri nets with 1 reset and 1 inhibitor arc (without hierarchy). 

\section*{Acknowledgments}
We thank Alain Finkel and Mohamed Faouzi Atig for insightful discussions and pointers to earlier results and reductions.
\bibliographystyle{plain}
\bibliography{papers}

\appendix
\section{Appendix}
 \label{app:main}

\subsection{Reduction from Reachability to Deadlockfreeness}
In this subsection, we shall discuss a reduction from Reachability in Petri nets to Deadlockfreeness in Petri nets. Further, this reduction only involves addition of classical arcs, and hence, this reduction holds for all extensions of Petri nets being considered in this paper.
\begin{figure}[h]
    \centering
        \begin{tikzpicture}
            \node[circle, draw] at (5,6) (P1) {$p_1$};
            \node[circle, draw] at (5,5) (P2) {$p_2$};
            \node[circle, draw] at (5,4) (P3) {$p_3$};
            \node[circle, draw] at (1,6) (P*) {$p_*$};
            \node[circle, draw] at (4,1) (PH) {$p_{\#}$};
            \node[circle, draw] at (9,4) (P**) {$p_{**}$};
            \node at (4,7.1) (N) {Original Net};
            \draw[dotted] (2,3) rectangle (6,7);
            \draw[fill] (2.9,5.65) rectangle node[xshift=0.2cm, yshift=-0.35cm] {t} (3.1,6.35);
            \draw[fill] (6.65,3.9) rectangle node[xshift=0.2cm, yshift=-0.35cm] {$t_3$} (7.35,4.1);
            \draw[fill] (6.65,4.9) rectangle node[xshift=0.2cm, yshift=-0.35cm] {$t_2$} (7.35,5.1);
            \draw[fill] (6.65,5.9) rectangle node[xshift=0.2cm, yshift=-0.35cm] {$t_1$} (7.35,6.1);
            \draw[fill] (8.9,5.65) rectangle node[xshift=0.2cm, yshift=0.55cm] {$t_{**}$} (9.1,6.35);
            \draw[fill] (1.65,1.9) rectangle node[xshift=2.2cm, yshift=-0.35cm] {$t_{\#}$} (6.35,2.1);
            \draw[-latex,thick] (P**) .. controls (8.5,5) .. (8.9,6);
            \draw[-latex,thick] (9.1,6) .. controls (9.5,5) .. (P**);
            \draw[-latex,thick] (P*) -- (2.9,6);
            \draw[-latex,thick] (3.1,6) .. controls (4,6) and (4,7) .. (P*);
            \draw[-latex,thick] (P1) .. controls (6,5.6) .. (7,5.9);
            \draw[-latex,thick] (7,6.1) .. controls (6,6.4) .. (P1);
            \draw[-latex,thick] (P2) .. controls (6,4.6) .. (7,4.9);
            \draw[-latex,thick] (7,5.1) .. controls (6,5.4) .. (P2);
            \draw[-latex,thick] (P3) .. controls (6,3.6) .. (7,3.9);
            \draw[-latex,thick] (7,4.1) .. controls (6,4.4) .. (P3);
            \draw[-latex,thick] (P*) .. controls (1,4) and (2,4) .. (2,2.1);
            \draw[-latex,thick] (P**) .. controls (9,3) and (6,3) .. (6,2.1);
            \draw[-latex,thick] (P1) .. controls (3.5,6) .. node[near start,xshift=0.15cm,yshift=-0.2cm] {\tiny $M(p_1)$} (3.5,2.1);
            \draw[-latex,thick] (P2) .. controls (4,5) .. node[near start,xshift=0.15cm,yshift=-0.4cm] {\tiny $M(p_2)$} (4,2.1);
            \draw[-latex,thick] (P3) .. controls (4.5,4) .. node[near start,xshift=0.3cm,yshift=-0.6cm] {\tiny $M(p_3)$} (4.5,2.1);
            \draw[-latex,thick] (4,1.9) -- (PH);
        \end{tikzpicture}
    \caption{Reduction from Reachability to Deadlockfreeness}
    \label{fig:reach_ddlk}
\end{figure}
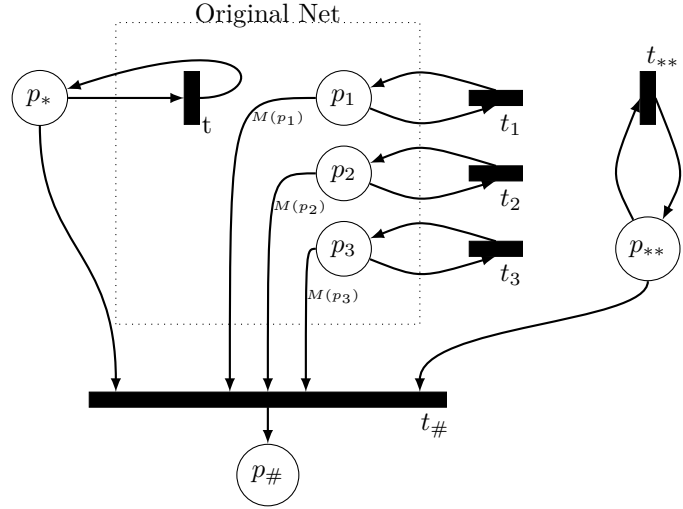
In Figure \ref{fig:reach_ddlk}, the construction has been described for reachability of marking $M$ in original net to deadlockfreeness in constructed net. 
\begin{itemize}
    \item We have added three places, $p_*$, $p_{**}$ and $p_{\#}$
    \item $p_*$ has a pre-arc and post-arc to every transition in the original net (as shown for transition $t$ in diagram)
    \item For every place $p_i$ in the original net, we add transition $t_i$ as shown in the diagram
    \item From each place $p_i$ we have a pre-arc of weight $M(p_i)$ to transition $t_\#$
    \item In initial marking in constructed net, $M'_0(p)=M_0(p)$ for all places in original net, $M'_0(p_*)=M'_0(p_{**})=1$ and $M'_0(p_{\#})=0$.
\end{itemize}
Let, for any marking $M$ in original net, define $M^{ext}$ in constructed net, such that, $M^{ext}(p)=M(p)$ for all places in original net, $M^{ext}(p_*)=M^{ext}(p_{**})=1$ and $M^{ext}(p_{\#})=0$. Thus, $M'_0=M_0^{ext}$.
\begin{lemma}\label{ddlk_forward}
If marking $M$ is reachable in the original net, then a deadlocked marking is reachable in the constructed net.
\end{lemma}
\begin{proof}
Given, $M$ is reachable in the original net. Let the corresponding run be $\rho$. Then, run $\rho$ is a run reaching $M^{ext}$ in the constructed net (this can be proved by induction over length of $\rho$).\\
Further, at marking $M^{ext}$, transition $t_\#$ is firable. On firing $t_\#$ at $M^{ext}$, we reach marking $M^{\star}$ in constructed net, where $M^{\star}(p_\#)=1$ and for any other place $p$, $M^{\star}(p)=0$.\\
But $M^{\star}$ as described above is a deadlocked marking. Thus, run $\rho t_\#$ reaches a deadlocked marking in the constructed net. Hence proved.
\end{proof}
\begin{lemma}\label{ddlk_mark}
$M^{\star}$ as described in Proof of Lemma \ref{ddlk_forward} is the only reachable deadlocked marking in the constructed net.
\end{lemma}
\begin{proof}
Consider any reachable deadlocked marking $M$. We shall prove that $M=M^{\star}$. Firstly, for all places $p_i$ present in the original net, $M(p_i)=0$, else transition $t_i$ would be firable. Similarly, $M(p_{**})=0$, else transition $t_{**}$ can fire. But we can show by induction that for any reachable marking $M_1$ $M_1(p_{**})=M_1(p_*)=1-M_1(p_\#)$. This implies that $M(p_*)=0$ and $M(p_\#)=1$. But then, $M=M^{\star}$.
\end{proof}
\begin{lemma}\label{ddlk_backward}
If a deadlocked marking is reachable in the constructed net, then marking $M$ can be reached in the original net.
\end{lemma}
\begin{proof}
Given, a deadlocked marking is reachable in the constructed net. By Lemma \ref{ddlk_mark}, the marking is $M^{\star}$. Let $\rho_1$ be the run reaching $M^{\star}$. Let $\rho$ be the run obtained from $\rho_1$ by removing the $t_i$ and $t_{**}$ transitions. This is again a valid run to $M^{\star}$ (since on firing $t_i$ or $t_{**}$, the marking remains unchanged). Now, consider any firing of $t_\#$ in the run $\rho$. In the marking reached after firing $t_\#$, there are zero tokens in $p_*$. Thus, no transition other than $t_i$ and $t_{**}$ is firable in this marking. This implies that $t_\#$ can only be the last transition fired in $\rho$. Also, since $M^{\star}(p_\#)=1$ (and only transition $t_\#$ puts tokens into $p_\#$), the last transition fired must be $t_\#$. Thus, $\rho$ is of the form $\mu t_\#$, where $\mu$ is a run over the transitions in the original net. Now, if we consider run $\mu$ in the original net, we notice that it must reach $M$, since $\mu$ must reach $M^{ext}$ in the constructed net. Thus, marking $M$ is reachable in the original net.
\end{proof}
\begin{theorem}\label{thm:folklore}
Reachability in Petri nets is reducible to Deadlockfreeness in Petri nets.
\end{theorem}

\subsection{Reduction from \Za{IR} to \Za{I}}
\label{app:hirpn_bisimulation}
We present a detailed proof of Lemma~\ref{lem:hirpnk-to-hirpnk-1} and Theorem~\ref{thm:bisim-k}.

\begin{lemma} \label{in-k-1}
The net $N'$ constructed in the proof of Lemma~\ref{lem:hirpnk-to-hirpnk-1}
is in the class \Za{IR}$_{k-1}$.
\end{lemma}

\begin{proof}
Since the reset arc from $p^R$ to $t$ in $N$ does not exist in the
net $N'$, we have at least one less reset transition than $N$ in $N'$.
Since $N \in \Za{IR}_k$, we get that $N$ has at most $k$ reset transitions. 
Hence, $N'$ has at most $k-1$ reset transitions. We further need
to establish that the hierarchy is preserved in $N'$. Consider $N'$ in
Figure~\ref{fig:transform-hirpn}. Since we only add a simple pre-arc
from $p^*$ to the transitions in the rest of the net, we preserve the hierarchy
in the rest of the net. Also, since $t^S$ and $t^R$ only
have simple pre-arcs, the hierarchy is preserved at both these transitions.
Finally, since in the original net $N$, we have special arcs from
$p^I$ and $p^R$ to $t$ which maintain the hierarchy, this hierarchy should
also be maintained at transition $t^I$ which also has special arcs
only from $p^I$ and $p^R$.
\end{proof}

\begin{lemma} \label{sum_p_p_t}
For any marking $M'$ reachable from $M'_0$ in $N'$, $M'(p^*) + M'(p_t^*)=1$.
\end{lemma}

\begin{proof}
A firing of transition $t^S$ will remove one token $p^*$ and add a token to $p_t^*$,
while a firing of $t^I$ will remove a token from $p_t^*$ and add a token to $p^*$.
Since firing $t^R$ will replenish the removed token from $p_t^*$ and firing a transition
from the rest of the net will replenish the removed token from $p^*$, the sum
of tokens in $p^*$ and $p_t^*$ will remain conserved. Since this sum is $1$ in the
initial marking, $M'(p^*) + M'(p_t^*)=1$.
\end{proof}
The main role of $p^*$ is to act as a driver for the rest of net $N'$ while the role
of $p_t^*$ is to act as a driver for each transition $t^R$ that arises from a reset arc
to $t$ in the original transition. Lemma~\ref{sum_p_p_t} shows that either $p^*$ has
1 token, in which case the transitions in the rest of net are firable, or $p_t^*$ has
1 token, in which case each transition $t^R$ is firable.

\begin{lemma} \label{forward}
Let there be $n$ tokens in place $p^{R}$ just before $t$ fires in net $N$. Then,
one firing of $t$ in $N$ is equivalent to the firing sequence containing
one firing of $t^S$, $n$ firings of $t^R$ and one firing of $t^I$ in that order
in the net $N'$.
\end{lemma}

\begin{proof}
Let $M_1 \Arrow{t} M_2$ in net $N$.
\begin{itemize}
\item Since $t$ is firable at $M_1$, $M_1(p^S) \geq 1$ and $M_1(p^I) = 0$.
Consider $M_1' = f(M_1)$. Clearly, $M_1'(p^*) = 1$ (by definition of $f$).
Hence, $t^S$ is firable in $N'$.

\item Once $t^S$ fires in $N$, $p^*$ becomes empty, hence the only
firable transitions are $t^R$ and $t^I$. If $M_1(p^{R}) = n > 0$, then $t^I$
cannot fire as it has an inhibitor arc from $p^{R}$. Hence, the only
transition that can fire until $p^{R}$ is empty is $t^R$. Essentially,
the $n$ firings of $t^R$ empty the place $p^{R}$ having the same effect
on it as transition $t$ in net $N$.

\item Once $p^{R}$ is empty, $t^I$ is the only transition that can fire,
emptying $p_t^*$ and putting 1 token back in $p^*$, signaling that the
transitions in rest of net can fire. Once $t^I$ fires, the marking in $N'$ is $M_2'
= f(M_2)$.
\end{itemize}
Hence, $f(M_1) \Arrow{{t^S \cdot (t^R)^n \cdot t^I}} f(M_2)$ in the net $N'$.
\end{proof}
This will serve as the basis for the \Za{IR}$_{k}$ $\to$ \Za{IR}$_{k-1}$
direction of the reduction.

\begin{lemma}[Forward Direction] \label{forward-sim}
For any markings $M_1$ and $M_2$ in $N$ such that $M_1 \Arrow{\mu} M_2$,
where $\mu$ is a sequence of transitions, there exists $\mu'$ such that
$f(M_1) \Arrow{{\mu'}} f(M_2)$ in net $N'$.
\end{lemma}

\begin{proof}
$\mu'$ can be constructed from $\mu$ by replacing each occurrence of $t$ in
$\mu$ by its equivalent sequence of transitions in $N'$ as demonstrated in
Lemma~\ref{forward}.
\end{proof}

\begin{lemma} \label{backward}
Consider markings $M_1'$ and $M_2'$ such that $M_1' \Arrow{{\mu'}} M_2'$
and $M_1'(p^*) = M_2'(p^*) = 1$
and $M_1'(p^{R}) = n$ in the net $N'$. If $\mu'$ starts with $t^S$, then
$t^S$ is followed by $n$ firings of $t^R$, followed by one firing of $t^I$
in $\mu'$.
\end{lemma}

\begin{proof}
Firing of $t^S$ places one token in $p_t^*$ guaranteeing the sequence
above as argued in Lemma~\ref{forward}.
\end{proof}

\begin{lemma}[Backward Direction] \label{backward-sim}
For any markings $M_1' = f(M_1)$ and $M_2' = f(M_2)$ in $N'$ such that
$M_1' \Arrow{{\mu'}} M_2'$,
where $\mu'$ is a sequence of transitions, there exists $\mu$ such that
$M_1 \Arrow{\mu} M_2$ in net $N$.
\end{lemma}

\begin{proof}
$\mu$ is constructed by replacing every
occurrence of $t^S \cdot (t^R)^n \cdot t^I$ by $t$. The existence of this
occurrence is guaranteed by Lemma~\ref{backward}.
\end{proof}
This concludes the proof of Lemma~\ref{lem:hirpnk-to-hirpnk-1}, establishing the relation between \Za{IR}$_k$ and \Za{IR}$_{k-1}$ by construction. 

\begin{figure}[t]
\begin{tikzpicture}[scale=0.75]
\hspace{0pt}\raisebox{-1em}
{
\draw (0,13.875) circle (0.3cm) node{\scriptsize $p^S$};
\draw (1,13.875) circle (0.3cm) node{\scriptsize $p^I$};
\draw (2,13.875) circle (0.3cm) node{\scriptsize $p^{R}$};
\draw (0,11.625) circle (0.3cm) node{\scriptsize $p$};
\draw (2,11.625) circle (0.3cm) node{\scriptsize $p^T$};
\draw [fill=black] (0.65,12.675) node[left]{\scriptsize $t$} rectangle (1.35,12.825);
\draw[-latex,thick] (0,13.575) -- (0.65,12.825);
\draw[-o,thick] (1,13.575) -- (1,12.825);
\draw[-latex,thick] (2,13.575) ..controls (1,12.675) .. node[right]{\scriptsize Transfer} (2,11.925);
\draw[-latex,thick] (1,12.675) -- (0,11.925);
\draw (3,12.75) rectangle node{\scriptsize \textit{Rest of Net}} (5,11.625);
}

\hspace{15pt}\raisebox{3em}
{
\path[draw=black,solid,line width=2mm,fill=black,
preaction={-triangle 90,thick,draw,shorten >=-1mm}
] (4.6, 10.875) -- (6.3, 10.875);
}

\hspace{180pt}\raisebox{9.3em}
{
\draw (0,10.125) circle (0.3cm) node{\scriptsize $p^S$};
\draw (2,10.125) circle (0.3cm) node{\scriptsize $p^*$};
\draw (1,7.875) circle (0.3cm) node{\scriptsize $p_t^*$};
\draw [fill=black] (0.65,8.925) node[left]{\scriptsize $t^S$} rectangle (1.35,9.075);
\draw[-latex,thick] (0,9.825) -- (0.8,9.075);
\draw[-latex,thick] (1.9,9.825) -- (1.1,9.075);
\draw[-latex,thick] (1,8.925) -- (1,8.175);
\draw (3,9.375) rectangle node{\scriptsize \textit{Rest of Net}} (5,8.25);
\draw[-latex,thick,dotted] (2.4,10.125) to[out=-20, in=100] (3.8, 9.15);
\draw[-latex,thick,dotted] (3.2,9.15) to[out=170, in=-60] (2,9.825);

\draw (2,6.75) circle (0.3cm) node{\scriptsize $p^{R}$};
\draw (5,6.75) circle (0.3cm) node{\scriptsize $p^{T}$};
\draw (0,6.75) circle (0.3cm) node{\scriptsize $p^I$};
\draw (1,4.5) circle (0.3cm) node{\scriptsize $p$};
\draw [fill=black] (3.5,6.6375) node[below]{\scriptsize $t^R$} rectangle (3.7,7.1625);
\draw [-latex,thick] (2.3,6.75) -- (3.5,6.75);
\draw [-latex,thick] (3.7,6.75) -- (4.7,6.75);
\draw [-latex,thick] (1,7.575) .. controls (1.5,7.275) and (2.5,7.125) .. (3.5,7.05);
\draw [-latex,thick] (3.7,7.05) .. controls (4.5,7.5) and (2.5,7.725) .. (1.3,7.875);
\draw [fill=black] (1.35,5.7) node[below right]{\scriptsize $t^I$} rectangle (0.65,5.55);
\draw[-o,thick] (0,6.45) -- (0.65,5.7);
\draw[-latex,thick] (1,7.575) -- (1,5.7);
\draw[-o,thick] (2,6.45) -- (1.35,5.7);
\draw[-latex,thick] (1,5.55) -- (1,4.8);
\draw[-latex,thick] (0.8,5.55) .. controls (-2,3.375) and (-3,13.5) .. (2,10.425);
}
\end{tikzpicture}
\vspace{-11.5em}
\caption{Transformation from $N \in$ \Za{IRcT}$_k$ (left) to $N' \in$ \Za{IRcT}$_{k-1}$ (right)}
\label{fig:transform-hirctpn}
\end{figure}

\begin{lemma} \label{reduction-k}
Reachability, coverability and termination problems in \Za{IR}$_{k}$ are reducible to their corresponding versions in \Za{IR}$_{k-1}$.
\end{lemma}

Finally, starting from  an arbitrary net $N \in \Za{IR}_k$. We apply Lemma~\ref{lem:hirpnk-to-hirpnk-1} recursively for each transition connected to a reset arc. The reduction follows from applying Lemma~\ref{reduction-k} successively.

The construction for constrained transfer arcs is shown in Figure~\ref{fig:transform-hirctpn}. Note that the constrained property of transfer arcs is required here, since if we had transfer to a place with an inhibitor arc to the same transition, then in the constructed net, $t_I$ cannot be fired, since we would have added tokens through $t_R$. Hence, we can redo the above formal proof and thus, as a consequence we obtain the proof of Theorem~\ref{thm:bisim-k}.


\subsection{Deadlockfreeness in \Za{IR}} \label{deadlock-appendix}
\begin{lemma} \label{4.1}
For any clause C, $[\forall i \forall p\in B^C_i, M(p)=i$ and $\forall p\in A^C, M(p)\geq 1]$ iff clause C is true at marking $M$.
\end{lemma}
\begin{proof}
If Clause C is true at marking M, by definition of $Exact_i$ and $AtLeast$, the result follows.\\
If $[\forall i \forall p\in B^C_i, M(p)=i$ and $\forall p\in A^C, M(p)\geq 1]$, then all literals of the form $Exact_i(p)$ in $S_C$ are true by definition of $B^C_i$. For all literals of the form $AtLeast(p)$ in $S_C$, we have $p\in A^C \bigcup_{i\geq 1}B^C_i$ by definition of $A^C$. Hence, $p\in A^C$ or $p\in B^C_i$ for some $i\geq 1$. If $p\in A^C$, then $M(p)\geq 1$. If $p\in B_i,i\geq 1$, then $M(p)=i\geq 1$. Thus $AtLeast(p)$ is true. Hence, clause C is true at marking $M$.
\end{proof}
\begin{lemma} \label{4.2}
Consider any run $\rho$ over $T$. $\rho$ is a run in the original net from $M_0$ to $M$, iff $\rho$ is a run in the constructed net from $M_0^{ext}$ to $M^{ext}$
\end{lemma}
\begin{proof}
Forward direction is trivial, since the constructed net has all transitions present in the original net unmodified.\\
For the backward direction, since $M^{ext}(p^*)=0$, for any clause C, transition $t_C$ was never fired in the run, since otherwise, a token would be added in $p^*$ which can't be removed by firing any other transition. Further, no other newly added transition was fired, since $M^{ext}(p_C)=0$ and all other new transitions add a token to $p_C$, which can be emptied only by $t_C$, which never fired. Hence, only transitions in the original net fired in the run. Thus, the run is a valid run in the original net too.
\end{proof}
\begin{lemma} \label{4.3}
Let marking $M$ be a deadlocked marking reachable from initial marking $M_0$ in original net. Then, marking $M'$, where $$M'(p)=\begin{cases} 1\, p=p^*\\ 0\, p\neq p^*\end{cases}$$ is reachable in constructed net.
\end{lemma}
\begin{proof}
Since $M$ is a deadlocked marking, $Deadlock(M)$ is true. This implies atleast one clause in the DNF is true. Let Clause C be any one of those clauses.\\
Let $\rho$ be the run from $M_0$ to $M$ in the original net. By Lemma \ref{4.2}, $\rho$ is a run from $M_0^{ext}$ to $M^{ext}$ in the constructed net. Since $M$ satisfies Clause C, we have $\forall i\forall p\in B^C_i M(p)=i$ and $\forall p\in A^C M(p)\geq 1$, by Lemma \ref{4.1}.
By definition of $M^{ext}$, $\forall i\forall p\in B^C_i M^{ext}(p)=i$ and $\forall p\in A^C M^{ext}(p)\geq 1$ (since $A^C$ and $B_i$ are subsets of $P$). Let $M^{ext}\Arrow{q_C}M_1$. Then, $\forall i\forall p\in B^C_i M_1(p)=i$ and $\forall p\in A^C M_1(p)\geq 1$. Consider the run defined as $$\rho'=\rho.q_C.\Bigcdot_{\forall j p_i\in B^C_j}(r_i).\Bigcdot_{p_i\in A^C}(t_{i*}^{M_1(p_i)-1}t_i).\Bigcdot_{p_i\notin A^C\cup\bigcup_{j}B^C_j}(s_i^{M(p_i)}).t_C$$ The marking reached by $\rho$ is $M^{ext}$. The marking reached by $\rho q_C$ is $M_1$. Now, $\Bigcdot_{\forall j p_i\in B^C_j}(r_i)$ removes $i$ tokens in all places $p_i\in B^C_i$, thus emptying the place and putting one token in $p_{i*}$. $\Bigcdot_{p_i\in A^C}(t_{i*}^{M_1(p_i)-1}t_i)$ removes all tokens from any place $p_i\in A^C$ and puts one token in $p_{i*}$. $\Bigcdot_{p_i\notin A^C\cup\bigcup_{j}B^C_j}(s_i^{M(p_i)})$ removes all tokens from any place $p\not\in A^C\cup \bigcup_{i}B^C_i$. In the resultant marking, $t_C$ is firable. Firing $t_C$ removes all tokens in all other places, and puts one token in $p^*$. Thus, $\rho'$ is a run in the constructed net, from $M_0^{ext}$ to $M'$. Hence proved.
\end{proof}
\begin{lemma} \label{4.4}
Let marking $M$ be reachable from $M_0^{ext}$ in the constructed net. Then, $$M(p^{**})+M(p^{*})+\Sigma_{check\, transition\, t_C}M(p_C)=1$$
\end{lemma}
\begin{proof}
Let $\rho$ be the run from $M_0^{ext}$ to $M$.
We shall prove this result by induction over the length $|\rho|$.
If $|rho|=0$, then $M=M_0^{ext}$ are the result follows.
Assume the statement to be true for $|rho|=k$.
Consider any $\rho$ of length $k+1$. Then, $\rho=\mu t$, where $\mu$ is a run of length $k$ and $t\in T'$. Let $M_0^{ext}\Arrow{\mu}M_1\Arrow{t}M$. By induction hypothesis, $M_1(p^{**})+M_1(p^{*}+\Sigma_{check\, transition\, t_C}M_1(p_C)=1$.\\
If $t=q_C$, then  $M_1(p^{**})=M(p^{**})$,$M_1(p^{*})-1=M(p^{*})$ and $1+\Sigma_{t_C}M_1(p_C)=\Sigma_{t_C}M(p_C)$.\\
If $t=t_C$ is a check transition, $1+M_1(p^{**})=M(p^{**})$,$M_1(p^{*})=M(p^{*})$ and $-1+\Sigma_{t_C}M_1(p_C)=\Sigma_{t_C}M(p_C)$.\\
Otherwise, $M_1(p^{**})=M(p^{**})$,$M_1(p^{*})=M(p^{*})$ and $\Sigma_{t_C}M_1(p_C)=\Sigma_{t_C}M(p_C)$.\\
Thus for all $t\in T'$, the summation remains constant. Hence proved.
\end{proof}
\begin{lemma} \label{4.5}
Run $\rho$ is a run from $M_0^{ext}$ to $M$, where $M(p^{*})=1$ iff $\rho=\mu t_C$, where $t_C$ is a check transition.
\end{lemma}
\begin{proof}
Consider a run $\rho$, such that no check transition $t_C$ is fired along $\rho$. We shall prove by induction over length of $\rho$ that if $M_0^{ext}\Arrow{\rho}M$, then $M(p^{*})\neq 1$.\\
If $|\rho|=0$, then marking reached by $\rho$ is $M_0^{ext}$, and by definition, $M_0^{ext}(p^{*})\neq 1$.\\
Assume it true for all $\rho$ of length $k$, such that $t_C$ is not fired in $\rho$ for all check transitions $t_C$.\\
Consider any $\rho$ of length $k+1$ such that no check transition is fired in it. Let $\rho=\mu t$, where $\mu$ is of length $k$ and has no $t_C$ fired in it. Hence, by induction hypothesis, if $M_0^{ext}\Arrow{\mu}M_1$, then $M_1(p^{*})\neq 1$. Since $t\neq t_C$ for all clauses C, hence, if $M_1\Arrow{t}M_2$, then $M_2(p^*)=M_1(p^*)\neq 1$ (Since only check transitions have pre-arc or post-arc to $p^*$). But $M_0^{ext}\Arrow{\mu}M_1\Arrow{t}M_2\implies M_0^{ext}\Arrow{\rho}M_2$. Hence proved.\\
Hence, taking contrapositive of this statement, we get, 
If run $\rho$ is a run from $M_0^{ext}$ to $M$, where $M(p^{*})=1$, then $\rho=\mu t_C \mu'$, where $t_C$ is a check transition. Now, we notice that $\mu'$ must be empty. This is because, if $M_0^{ext}\Arrow{{\mu t_C}}M_3$, then $M_3(p^*)=1$. By Lemma \ref{4.4}, $M_3(p^{**})=M_3(p_C)=0$ for all clauses C. Hence, in marking $M_3$, no transition is firable in the constructed net. Hence, $\mu'$ must be the empty run. Thus, we have proved the forward direction of the claim.\\
The other direction is trivial.
\end{proof}
\begin{lemma} \label{4.6}
If run $\rho$ is a run from $M_0^{ext}$ to $M$, where $M(p_C)=1$, then $q_C$ fires in run $\rho$.
\end{lemma}
\begin{proof}
Consider a run $\rho$, such that no transition $q_C$ is fired along $\rho$. We shall prove by induction over length of $\rho$ that if $M_0^{ext}\Arrow{\rho}M$, then $M(p_C)=0$.\\
If $|\rho|=0$, then marking reached by $\rho$ is $M_0^{ext}$, and by definition, $M_0^{ext}(p^{*})=0$.\\
Assume it true for all $\rho$ of length $k$, such that $q_C$ is not fired in $\rho$.\\
Consider any $\rho$ of length $k+1$ such that transition $q_C$ is not fired in it. Let $\rho=\mu t$, where $\mu$ is of length $k$ and has no $q_C$ fired in it. Hence, by induction hypothesis, if $M_0^{ext}\Arrow{\mu}M_1$, then $M_1(p^{*})=0$. Since $t\neq q_C$, hence, if $M_1\Arrow{t}M_2$, then $M_2(p^*)=M_1(p^*)=0$ (Since only transition $q_C$ can add tokens to $p_C$). But $M_0^{ext}\Arrow{\mu}M_1\Arrow{t}M_2\implies M_0^{ext}\Arrow{\rho}M_2$. Hence proved.\\
\end{proof}
\begin{lemma} \label{4.7}
If transition $q_C$ fires in a run $\rho$, then the $\rho=\mu q_C \mu'$, where $\mu'$ is a run over Clause C Transitions and $\mu$ is a run over $T$.
\end{lemma}
\begin{proof}
Let $M_0^{ext}\Arrow{\mu}M_1\Arrow{q_C}M_2\Arrow{{\mu'}}M_3$.\\
Firstly, we shall prove by induction on length of run $\rho$ that if $\rho$ is a run from $M_2$ from $M$, then $M(p_C)+M(p^*)=1$.
If $|\rho|=0$, then $M=M_2$. Since $q_C$ has a post-arc of weight 1 to $p_C$, $M(p_C)=M_2(p_C)=1$ and $M(p^*)=M_2(p^*)=0$. Hence, the hypothesis holds.\\
Assume true for all $|\rho|=k$.\\
Consider $\mu$ of length $k+1$. Let $\mu=\rho t$, where $|\rho|=k$ and $t\in T'$. By induction hypothesis, if $M_2\Arrow{\rho}M_4\Arrow{t}M_5$, then $M_4(p_C)+M_4(p^*)=1$.\\
If $t=t_C$, then $M_5(p_C)=M_4(p_C)-1$ and $M_5(p_C)=M_4(p_C)+1$
Otherwise, $M_5(p_C)=M_4(p_C)$ and $M_5(p_C)=M_4(p_C)$
Hence, the sum remains same. Hence proved.\\
Now, consider any marking $M$ reachable from $M_2$. Then, $M(p_C)+M(p^*)=1$. If $M(p^*)=1$, by Lemma \ref{4.4}, no transition can be fired at $M$. Otherwise, $M(p_C)=1$. Then, Clause C Transitions are firable. Hence, from any marking reachable from $M_2$, only Clause C Transitions can be fired. This implies $\mu'$ must be a run over Clause C Transitions.\\
Now, we need to prove that $\mu$ must be a run over $T$. Assume that some $t\not\in T$ fires in $\mu$ at marking $M$. Then, we must have $M(p_{C'})=1$ for some clause $C'$. Then, transition $q_{C'}$ must have fired in $\mu$ by Lemma \ref{4.6}. But, then, after $q_{C'}$ fired, $q_C\not \in$ Clause $C'$ Transitions fired. Contradiction.
Hence Proved.
\end{proof}
\begin{lemma} \label{4.8}
Consider any marking $M$ reachable from marking $M_0^{ext}$, where $M(p_C)=1$. Then, for any marking $M'$ reachable from $M$,s.t. $M'(p_C)=1$,
\begin{itemize}
    \item $\forall i\forall p_j\in B^C_i M'(p_j)+i*M'(p_{j*})=M(p_j)+i*M(p_{j*})$
    \item $\forall p_i\in A^C M'(p_i)+M'(p_{i*})\leq M(p_i)+M(p_{i*})$
\end{itemize}
\end{lemma}
\begin{proof}
Given, $M(p_C)=1$. By Lemma \ref{4.6}, $q_C$ must fire in the run $\rho$ from $M_0^{ext}$ to $M$. By Lemma \ref{4.7}, the run from $M$ to $M'$, say $\mu$ is over Clause C Transitions.\\
Since $M'(p_C)=1$, $t_C$ is not fired in $\mu$ by Lemma \ref{4.4} and Lemma \ref{4.5}. Hence, the only transitions firing in $\mu$ are $r_i,s_i,t_i,t_{i*}$.
Both the results can be proved by induction over the length of $\rho$, similar to Lemma \ref{4.4}.
\end{proof}
\begin{lemma} \label{4.10}
Let marking $M'$, as defined in Claim \ref{4.3}, be reachable from $M_0^{ext}$ in the constructed net. Then, there exists a marking $M$ reachable from $M_0$ in the original net such that marking $M$ satisfies $Deadlock(M)$.
\end{lemma}
\begin{proof}
Let $\rho$ be the run from $M_0^{ext}$ to $M'$ in the constructed net. Since $M'(p^*)=1$, by Lemma \ref{4.5}, $\rho=\mu t_C$ for some check transition $t_C$. Let $M_0^{ext}\Arrow{\mu}M_1\Arrow{t_C}M'$. Then, by definition of the constructed net, $\forall p_i\in A^C,M_1(p_{i*})=1$ and $\forall j\forall i \in B^C_j,M_1(p_{i*})=1$ and $\forall p\in P, M_1(p)=0$ and $M_1(p_C)=1$ (for $t_C$ to be firable at $M_1$). By Lemma \ref{4.6} and \ref{4.7}, $\mu=\alpha q_C\beta$, where $\beta$ is a word over Clause C Transitions and $\alpha$ is a word over $T$. Consider $M_0^{ext}\Arrow{\alpha}M_2\Arrow{q_C}M_3\Arrow{\beta}M_1\Arrow{t_C}M'$. We have, $\forall i,M_2(p_{i*})=0$ since transitions $t\in T$ don't add any tokens to $p_{i*}$. Hence, $\forall i,M_3(p_{i*})=0$ and $M_3(p_C)=1$. By Lemma \ref{4.8}, 
\begin{itemize}
    \item $\forall i\forall p_j\in B^C_i M_1(p_j)+i*M_1(p_{j*})=M_3(p_j)+i*M_3(p_{j*})$
    \item $\forall p_i\in A^C M_1(p_i)+M_1(p_{i*})\leq M_3(p_i)+M_3(p_{i*})$
\end{itemize}
This further implies,
\begin{itemize}
    \item $\forall i\forall p_j\in B^C_i i=M_3(p_j)$
    \item $\forall p_i\in A^C 1\leq M_3(p_i)$
\end{itemize}
Thus, $M_3$ satisfies clause C by Lemma \ref{4.1}. Hence, $M_3$ satisfies $Deadlock(M)$. This further implies that $M_2$ satisfies $Deadlock(M)$ (since $M_3(p)=M_2(p)\forall p\in P$) Also, $M_2=M^{ext}_{org}$ for some $M_{org}$, since $M_2(p^{**})=1$. Then, by Lemma \ref{4.2}, $M_{org}$ is reachable from $M_0$ in the original net and also satisfied $Deadlock(M)$. Hence, $M_{org}$ is a deadlocked marking reachable from $M_0$. Hence proved.
\end{proof}
This proves that Deadlockfreeness of the original \Za{IR} net is equivalent to reachability of marking $M'$ in the constructed \Za{IR} net.

\subsection{Reachability in Petri nets with 1 reset arc and 1 inhibitor arc}\label{1I1RPN}
Here, we present a reduction from Reachability in Petri Nets with two inhibitor arcs to Petri Nets with one reset and one inhibitor arc.

The construction is shown in the diagram alongside.

\begin{figure}[t]
\begin{tikzpicture}
\draw (0,18.5) circle (0.3cm) node{$p_1$};
\draw (2,18.5) circle (0.3cm) node{$p_2$};
\draw (1,15.5) circle (0.3cm) node{$p_3$};
\draw [fill=black] (0.65,16.9) node[left]{$t_1$} rectangle (1.35,17.1);
\draw [fill=black] (2.65,16.9) node[left]{$t_2$} rectangle (3.35,17.1);
\draw [fill=black] (4.65,16.9) node[left]{$t_3$} rectangle (5.35,17.1);
\draw[-latex,thick] (3,17.1) -- (2,18.2);
\draw[-latex,thick] (2,18.2) -- (5,17.1);
\draw[-latex,thick] (0,18.2) -- node[left]{1} (0.65,17.1);
\draw[-o,thick] (2,18.2) -- (1.35,17.1);
\draw[-latex,thick] (1,16.9) -- node[left]{1} (1,15.8);
\draw (0,14) -- (5,14);
\draw (0,12.5) circle (0.3cm) node{$p_1$};
\draw (2,12.5) circle (0.3cm) node{$p_2$};
\draw (4,12.5) circle (0.3cm) node{$p'_2$};
\draw (1,9.5) circle (0.3cm) node{$p_3$};
\draw [fill=black] (0.65,10.9) node[left]{$t_1$} rectangle (1.35,11.1);
\draw [fill=black] (2.65,10.9) node[left]{$t_2$} rectangle (3.35,11.1);
\draw [fill=black] (4.65,10.9) node[left]{$t_3$} rectangle (5.35,11.1);
\draw[-latex,thick] (2.65,11.1) -- node[right]{1} (2,12.2);
\draw[-latex,thick] (2,12.2) -- node[left]{1} (4.65,11.1);
\draw[-latex,thick] (3.35,11.1) -- node[right]{1} (4,12.2);
\draw[-latex,thick] (4,12.2) -- node[left]{1} (5.35,11.1);
\draw[-latex,thick] (0,12.2) -- node[left]{1} (0.65,11.1);
\draw[-latex,thick] (2,12.2) -- node[left]{rt} (1.35,11.1);
\draw[-latex,thick] (1,10.9) -- node[left]{1} (1,9.8);
\end{tikzpicture}
\end{figure}
In the construction, for one place (place $p_2$ in the diagram) from which we have an inhibitor arc, we create a copy place (place $p'_2$ in diagram). In the initial marking, the copy place has equal number of tokens as the original place. The copy place has the same set of arcs as the original place, except for the inhibitor arc. Also, the inhibitor arc of original place is now replaced by a reset arc. The intuitive idea is that the if reset arc is fired when the original place has non-zero number of tokens, then the number of tokens in original place and copy place will not be equal. This can be checked by reachability.

For convenience, assume that the two inhibitor arcs are, one from place $p_2$ to transition $t_1$, and another, from any place to transition $t_4$.
Let the original net be $(P,T,F)$. Then the constructed net is $(P'=P\cup\{p_2'\},T,F')$. The flow function is $$F'(p,t)=\begin{cases}R\,\,\,\,\,\,\,\,\,\,\,\,\,\,\,\,\,\,\,p=p_2\wedge t=t_1\\0\,\,\,\,\,\,\,\,\,\,\,\,\,\,\,\,\,\,\,\,p=p'_2\wedge t=t_1\\0\,\,\,\,\,\,\,\,\,\,\,\,\,\,\,\,\,\,\,\,p=p'_2\wedge t=t_4\wedge F(p_2,t_4)=I\,(i.e.\, both\, inh.\, arcs\, from\, p_2)\\F(p_2,t)\,\,\,\,\,\,\,\,\,\,p=p'_2\wedge F(p_2,t)\in \mathbb{N}\\F(p,t)\,\,\,\,\,\,\,\,\,\,\,otherwise\end{cases}$$ $$\forall p\in P\,F'(t,p)=F(t,p)$$ $$F'(t,p'_2)=F(t,p_2)$$ Note that by construction, $\forall t\in T\,F'(p'_2,t)\in \mathbb{N}$, and  $\forall t\in T\,F'(p'_2,t)>0\implies F'(p'_2,t)=F(p_2,t)$.\\
For convenience, define for any marking $M$ in original net, $M^{ext}$ is a run in the constructed net, where $$M^{ext}(p'_2)= M(p_2)\wedge\forall p\in P\,M^{ext}(p)= M(p)$$
The initial marking in the constructed net is $M^{ext}_0$, where $M_0$ is the initial marking of the original net.
\begin{lemma} \label{4.21}
Consider a marking $M_1$ of the constructed net, such that $M_1(p_2)<M_1(p'_2)$. Then, for all markings $M_2$ reachable from $M_1$, $M_2(p_2)<M_2(p'_2)$.
\end{lemma}
\begin{proof}
Let the run from $M_1$ to $M_2$ be $\rho$.\\
We shall prove this result by induction over length of $\rho$.
If $|\rho|=0$, then $M_2=M_1$ and hence $M_2(p_2)=M_1(p_2)<M_1(p'_2)=M_2(p'_2)$.
Assume the statement to be true for all runs $\mu$ of size $k$.\\
Consider a run $\rho$ of size $k+1$. Let $\rho=\mu t$ where $\mu$ is a run of length $k$ and $t\in T$. Also, let $M_1\Arrow{\mu}M_3\Arrow{t}M_2$.\\
By induction hypothesis, we have $M_3(p_2)<M_3(p'_2)$.\\
If $F'(p_2,t)\in \mathbb{N}$, then $M_2(p_2)-M_3(p_2)=M_2(p'_2)-M_3(p'_2)=F'(t,p_2)-F'(p_2,t)$. Hence, $M_2(p'_2)> M_2(p_2)$.\\
If $F'(p_2,t)=I$, (i.e. $t=t_4$ and both inhibitor arcs from $p_2$ in original net), then $M_2(p_2)=F'(t,p_2)$ and $M_2(p'_2)=M_3(p'_2)+F'(t,p_2)>M_2(p_2)$.\\
If $F'(p_2,t)=R$, (i.e. $t=t_1$), then $M_2(p_2)=F'(t,p_2)$ and $M_2(p'_2)=M_3(p'_2)+F'(t,p_2)>M_2(p_2)$.
Thus, in all cases, we have $M_2(p_2)\leq M_2(p'_2)$.\\
\end{proof}
\begin{lemma} \label{4.22}
For all reachable markings $M_1$, we have $M_1(p_2)\leq M_1(p'_2)$.
\end{lemma}
\begin{proof}
The proof is similar to that of Lemma \ref{4.21}.
\end{proof}
\begin{lemma} \label{4.23}
Reachability in Inhibitor Petri Nets is reducible to Reachability in Reset Petri Nets. Also, this reduction conserves the hierarchy in the nets.
\end{lemma}
\begin{proof}
We shall prove that reachability of a marking $M$ in the original net is equivalent to the reachability of marking $M'=M^{ext}$.\\
Let the marking $M$ be reachable in the original net. Let the corresponding run be $\rho$. Consider the run in the constructed net. We claim that the run is valid and the marking reached is $M'$.\\
We shall show this by induction over length of the run.\\
If $|\rho|=0$, then $M=M_0$ and hence $M'=M^{ext}_0$. Then $\rho$ is a run from $M^{ext}_0$ to $M'$.\\
Assume the statement to be true for all runs $\mu$ of length $k$.\\
Consider any run $\rho$ of length $k+1$. Let $\rho=\mu t$, where $\mu$ is a run of length $k$ and $t\in T$. Also, let $M_0\Arrow{\mu}M_1\Arrow{t}M$.\\
By induction hypothesis, we have in the constructed net, $M^{ext}_0\Arrow{\mu}M^{ext}_1$. Given that transition $t$ fires at $M_1$. Thus, we have, $\forall p\in P$, if $F(p,t)\in \mathbb{N}$, then $M_1(p)>F(p,t)$, and if $F(p,t)=I$, then $M_1(p)=0$.\\
Consider any $p\in P\cup\{p'_2\}$.\\
If $t\neq t_1$, then $\forall p\in P\,F'(p,t)=F(p,t)$ by definition, and hence, $M_1^{ext}(p)=M_1(p)>F(p,t)=F'(p,t)$. Also, $F'(p'_2,t)\in \mathbb{N}$ by construction. If $F'(p'_2,t)>0$, then $F'(p'_2,t)=F(p_2,t)$, by construction. Now, by Lemma \ref{4.22}, $M_1^{ext}(p'_2)\geq M_1^{ext}(p_2)=M_1(p_2)\geq F(p_2,t)=F'(p'_2,t)$. Hence $t$ is firable at $M_1^{ext}$.\\
If $t=t_1$, then $F'(p_2,t)=R$,$F'(p'_2,t)=0$ and $\forall p\in P-\{p_2\}\,F'(p,t)=F(p,t)$. Again, $t$ is firable at $M_1^{ext}$.\\
Hence $t$ is firable at marking $M^{ext}_1$. Let $M^{ext}_1\Arrow{t}M'$.\\ \\
We already know that $M_1\Arrow{t}M$. We now need to prove that $M'=M^{ext}$.\\
If $t\neq t_1$:\\
\tab $\forall p\in P\,M'(p)=M_1^{ext}(p)-F'(p,t)+F'(t,p)=M_1(p)-F(p,t)+F(t,p)=M(p)$\\
\tab $M'(p'_2)=M_1^{ext}(p'_2)-F'(p'_2,t)+F'(t,p'_2)=M_1(p_2)-F(p_2,t)+F(t,p_2)=M(p_2)$\\
\tab Thus, $M'=M^{ext}$\\
If $t=t_1$:\\
\tab $\forall p\in P-\{p_2\}\,M'(p)=M_1^{ext}(p)-F'(p,t)+F'(t,p)=M_1(p)-F(p,t)+F(t,p)=M(p)$\\
\tab $M_1(p_2)=0$ for $t_1$ to be firable at $M_1$\\
\tab Then, $M'(p_2)=F'(t_1,p_2)=M_1(p_2)+F'(t_1,p_2)=M_1(p_2)+F(t_1,p_2)=M(p_2)$\\
\tab $M'(p_2)=M_1^{ext}(p'_2)+F'(t_1,p'_2)=M_1(p_2)+F(t_1,p_2)=M(p_2)$\\
\tab Thus, $M'=M^{ext}$\\
Hence proved.\\ \\
Let the marking $M'$ be reachable in the modified net. Thus, there exists a run $\rho$ from initial marking $M^{ext}_0$ to $M'$.\\
Claim, run $\rho$ is a valid run in the original net from $M_0$ to $M$.\\
We prove this by induction over length of $\rho$.\\
If $|\rho|=0$, then $M'=M^{ext}_0$, and then $M=M_0$ and hence $\rho$ is a run from $M_0$ to $M$.\\
Assume the statement to be true for all runs $\mu$ of length k.\\
Consider any run $\rho$ of length $k+1$. Let $\rho=\mu t$ for $\mu$ a run of length $k$, and $t\in T$. Also, let $M^{ext}_0\Arrow{\mu}M'_1\Arrow{t}M^{ext}$.\\
From Lemma \ref{4.22} and \ref{4.21}, we conclude that for all $p\in P$, $M'_1(p_2)=M'_1(p'_2)$. Thus, $M'_1=M^{ext}_1$ for some marking $M_1$. Also, if $t=t_1$, then $M_1(p_2)=0$. Else, the marking reached on firing $t$ cannot be $M^{ext}$ (Since $M^{ext}(p'_2)\neq M^{ext}(p_2)$ as we lose tokens from $p_2$ on firing $t_1$).\\
By induction hypothesis, $\mu$ is a run from $M_0$ to $M_1$.\\
Also, $F(p,t)=I$ implies $M_1(p)=0$ (If $(p,t)=(p_2,t_1)$, then discussed above, and otherwise, $M_1(p)=M_1^{ext}(p)=0$ since $F'(p,t)=F(p,t)=I$ and t is firable).\\
If $F(p,t)\in \mathbb{N}$, then $F'(p,t)\leq M^{ext}_1(p)$. Therefore, $F(p,t)\leq M_1(p)$.\\
Hence, $t$ can be fired at marking $M_1$. Let $M_1\Arrow{t}M^*$.\\
Now, we need only prove that $M^*=M$.
This is easy to see, as the flow relation is unchanged for the the set of places $P$. $M^*(p)=M_1(p)-F(p,t)+F(t,p)=M^{ext}_1(p)-F'(p,t)+F'(t,p)=M^{ext}(p)=M(p)$.\\
Hence, $M^*=M$.
\end{proof}
This completes the proof.
\begin{theorem}
Reachability in Petrinets with 1 inhibitor arc and 1 reset arc is undecidable
\end{theorem}
This subsection proved that reachability in undecidable in Petri nets with 1 inhibitor arc and 1 reset arc. This fixes the decidability frontier perfectly. While reachability in Petri nets with one inhibitor arc and Petri nets with one reset arc are decidable, reachability in Petri nets with one inhibitor arc and one reset arc is undecidable.\\
We can further provide a reduction from Deadlockfreeness in Petri nets with two inhibitor arcs (from two different places) to deadlockfreeness in Petri nets with one reset arc and hierarchical inhibitor arcs. For this,  we have two copy places, $p'_2$ and $p''_2$, instead of one (flow relation is same for both copy places). To above construction (call it N), we add a place $p^*$, which has a pre-arc and post-arc to all transitions in N. We add another place $p^{**}$ and a transition $t^*$, with pre-place as $p^*$ and post-place as $p^{**}$ with weight of all arcs as 1. For every transition $t$ in N, we add a new transition $t'$, such that $F(t',p_2)=F(t',p'_2)=F(p_2,t')=F(p'_2,t')=0\wedge\forall p\not \in \{p_2,p'_2\}\,F(t',p)=F(p,t')=F(p,t)$ (if $F(p,t')\not \in \mathbb{N}$ then $F(t',p)=0$). We add two more transitions, say $s^*$ and $r^*$, where $s^*$ has a pre-arc from $p_2$, $p'_2$ and $p^{**}$ each of weight 1, and a post-arc of weight 1 to $p^{**}$. Thus, $s^*$ just empties one token each from $p_2$ and $p'_2$ when $p^{**}$ has a token. Transition $r^*$ has a pre-arc from $p'_2$ and $p^{**}$ and post-arcs to $p'_2$ and $p^{**}$.\\
 Clearly, if a deadlocked marking is reachable in original net, then in the constructed net, one can reach that marking, then fire $t^*$, and empty $p'_2$ using $s^*$. Now, no transition is firable, and we reach a deadlocked marking in the constructed net.\\
 Conversely, if a deadlocked marking is reached in constructed net, then we must have a token in $p^{**}$ and $t^*$ has fired in the run (say the run is $\mu t^* \mu'$) (else, there is a token in $p^*$ and $t^*$ is firable). This implies that $p'_2$ has no tokens, which means number of tokens in $p_2$ (which must be lesser than number of tokens in $p'_2$ by Lemma \ref{4.21} above), must have 0 tokens, implying that $t_1$ was never fired when $p_2$ had non-zero tokens (By Lemma \ref{4.22}) and number of tokens in $p''_2$ and in $p_2$ are equal before $t^*$ fired. Since none of the other added transitions is firable, though number of tokens in $p''_2$ is equal to number of tokens in $p_2$, the marking reached by $\mu$ in original net is in deadlock.
 
\subsection{Coverability in \Zb{R}{I}}\label{2R1C}

First, we give proofs of the lemmas stated in section \ref{Cover}
\begin{lemma}
The place $q_{i}$ and only $q_{i}$ in the Petri net gets a token when the instruction numbered $i$ is being simulated.
\end{lemma}

\begin{proof}
The proof of this property of the run in the Petri net is by induction on the number of instructions simulated so far.

\textbf{Base Case} We start the Petri net with one token in $q_{0}$ and no tokens in all other places. So, when the first instruction is simulated, we have a single token in $q_{0}$ and no tokens in any other $q_{i}$

\textbf{Induction Hypothesis} Assume that when we are simulating the $k$th instruction of two counter machine which corresponds to instruction $i$ and hence we have a token in $q_i$. 

\textbf{Inductive Proof}  Now, consider the instruction number $k+1$. The instruction can be of two types.
  \begin{enumerate}
  \item INC instruction - From our construction, only one transition can fire and when it fires, $q_i$ becomes empty and $q_{j}$ gets the token.
  \item JZDEC instruction - From our construction, when $q_i$ gets a token, the place $S_{r1}$ gets a token. Now, two transitions can fire. Both of them use $S_{r1}$. Hence, only one of them can fire. As both of them use $q_{i1}$, it gets empty after the next transition. Hence, all other places in the net stay empty. 
  \end{enumerate}
\end{proof}

\begin{lemma} In any reachable marking $M$, $M(S) \geq M(C_{1})+M(C_{2})$ and $M(S) = M(C_{1})+M(C_{2})$ iff there are no incorrect transitions fired in the run to $M$\end{lemma}
\begin{proof}
Consider f = $M(S)-M(C_{1})-M(C_{2})$. If the transition is not incorrect, then $f$ remains the same. Because, if $S$ is incremented, $C_{1}$ or $C_{2}$ is incremented. Same is the case for decrement operation. But, if the transition is incorrect, $S$ is unchanged while $C_{1}$ or $C_{2}$ decreases. So, $f$ increases. Initially, $S$, $C_{1}$ and $C_{2}$ are empty, $f$ is zero. We prove that $f \geq 0$ 
\begin{itemize}
\item \textbf{Case 1a.} $C_{r}$ is non empty and the token correctly goes to place $q_{j}$ 
In this case, both $S$ and $C_{r}$ are decremented by 1. So, $f$ remains the same. 
\item \textbf{Case 1b.} $C_{r}$ is non empty and the token incorrectly goes to place $q_{l}$ In this case, $C_{r}$ is emptied while $S$ remains the same. Thus, $f$ increases. 
\item \textbf{Case 2a.}$C_{r}$ is empty and the token goes to place $q_{j}$ This case is not possible as the transition on the left (Guessing $C_{r}$ is non empty) can not fire with $C_{r}$ being empty.
\item \textbf{Case 2b.} $C_{r}$ is empty and the token correctly goes to place $q_{l}$ In this case, first $S_{r1}$ gets the token, the reset arc transition fires. $S_{r2}$ gets the token. Now, the transition that uses it fires, thus emptying $q_{i1}$. The token is present only in place $q_{l}$. $f$ doesn't change as the number of tokens in $S$ and $C_{r}$ remain the same. 
\end{itemize}
Notice that $f$ remains same in all the cases except in 1b where it's increased. This proves the lemma.
\end{proof}

\subsubsection{Checking for Incorrect Transitions}
 To check if $M(S)=M(C_1)+M(C_2)$, we do three steps. As a first step, We divide $q_{n}$ into three places $q_{n1}$, $q_{n2}$ and $q_{n3}$. 
\begin{center}
\includegraphics[width=5cm, height=4cm]{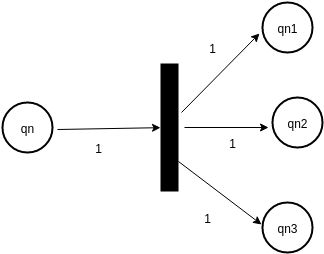} 
\end{center}
Step - 1 :
We transfer all the tokens from $C_{1}$ to $C_{2}$.  
\begin{center}
\includegraphics[width=5cm, height=4cm]{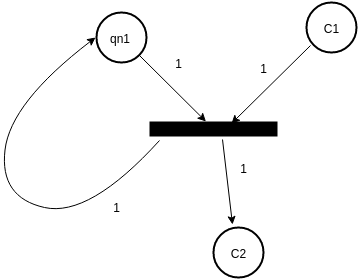} 
\end{center}
Step - 2 :
Remove $M(C_{2})$ number of tokens from $S$. 
\begin{center}
\includegraphics[width=5cm, height=4cm]{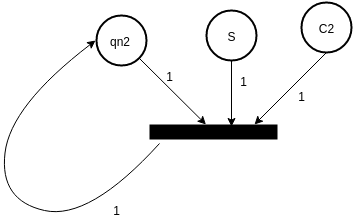} 
\end{center}
Step - 3 :
Now, all that remains to be checked is whether $S$ is empty or not. 
This can be done with an inhibitor arc.
\begin{center}
\includegraphics[width=5cm, height=4cm]{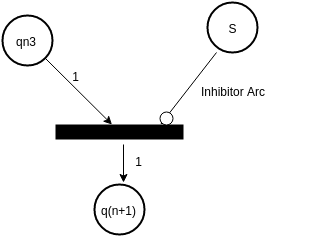} 
\end{center}

\subsubsection{Transfer Arcs and Place reachability}
The above proof works for 2 transfer arcs + 1 inhibitor arcs as well. We replace reset arc with transfer arc from $C_r$ to a dump place. Hence, Coverability in a Petri net with two transfer arcs and one inhibitor arc is undecidable. 
\begin{center}
\includegraphics[width=12cm, height=12cm]{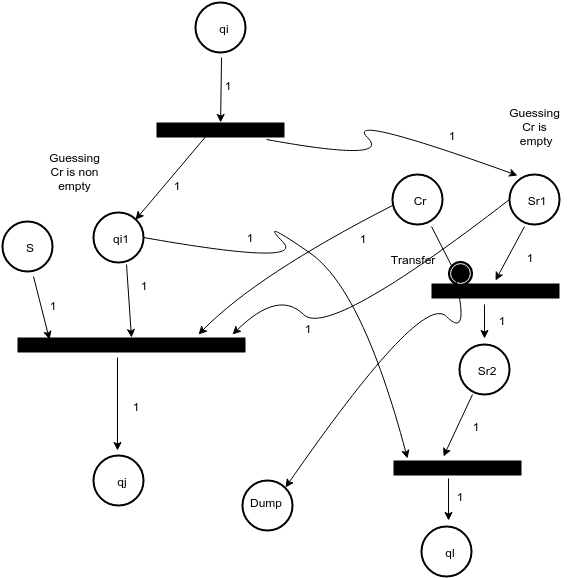} 
\end{center}
Decidability of Coverability in Petri  nets with 1 reset arc and 1 inhibitor arc is still open. Similarly, deadlockfreeness in Petri  nets with 1 reset arc and 1 inhibitor arc is open. In summary, when we add reset arcs without hierarchy to \Za{I}, termination remains decidable, while reachability, deadlockfreeness and coverability become undecidable. 

\subsection{Reachability in \Za{T}}\label{transfer_undec-appendix}
The construction discussed in \ref{transfer_undec} requires that there is no pre-arc from $p_1$ to $t_2$. Even if this is not the case, we can construct an equivalent petri net, with this property as shown in the following construction.

\begin{changemargin}{-1.5cm}{-1.5cm}
\begin{minipage}{.5\textwidth}
\begin{tikzpicture}[scale=0.7]
\draw (1,3) circle (0.3cm) node{$p_2$};
\draw (1,1) circle (0.3cm) node{$p_1$};
\draw (5,3) circle (0.3cm) node{$p_4$};
\draw (5,1) circle (0.3cm) node{$p_5$};
\draw[fill=black] (2.9,1.65) node[below]{$t_2$} rectangle (3.1,2.35);
\draw[-latex,thick] (1.3,1) -- (2.9,2);
\draw[-latex,thick] (3.1,2) -- (4.7,1);
\draw[-latex,thick] (1.3,3) .. controls (2.3,2) and (3.7,2) .. node[above]{tf} (4.7,3);
\draw[-,thick] (6,1) -- (6,3);
\draw (7,1) circle (0.3cm) node{$p_1$};
\draw[fill=black] (8.9,1.65) node[below]{$t_{2a}$} rectangle (9.1,2.35);
\draw (11,1) circle (0.3cm) node{$p_\#$};
\draw (7,3) circle (0.3cm) node{$p_{\$}$};
\draw (11,3) circle (0.3cm) node{$p_2$};
\draw (15,3) circle (0.3cm) node{$p_4$};
\draw (15,1) circle (0.3cm) node{$p_5$};
\draw[-latex,thick] (7.3,1) -- (8.9,2);
\draw[-latex,thick] (7.3,3) -- (8.9,2);
\draw[-latex,thick] (9.1,2) -- (10.7,1);
\draw[fill=black] (12.9,1.65) node[below]{$t_{2b}$} rectangle (13.1,2.35);
\draw[-latex,thick] (13.1,2) .. controls (23,4) and (8,4) .. (7.3,3);
\draw[-latex,thick] (11.3,1) -- (12.9,2);
\draw[-latex,thick] (13.1,2) -- (14.7,1);
\draw[-latex,thick] (11.3,3) .. controls (12.3,2) and (13.7,2) .. node[above]{tf} (14.7,3);
\end{tikzpicture}
\end{minipage}
\end{changemargin}\\ \\
Additionally, the place $p_{\$}$ has a pre-arc and post-arc to every transition other than $t_2$.
The above two nets are equivalent, and the right net has no newly introduced deadlocked reachable marking too.
\begin{lemma} \label{4.24}
For any reachable marking $M$, we have $M(p_*)=0\implies M(p_1)=0$. Similarly, $M(p'_*)=0\implies M(p'_1)=0$
\end{lemma}
\begin{proof}
This can be shown by induction over the length of run $\rho$ to $M$ from $M'_0$.
\end{proof}
\begin{lemma} \label{4.25}
If marking $M$ is reachable from $M_0$ in the original net, then $A_{M}$ or $B_{M}$ is reachable from $M'_0$.
\end{lemma}
\begin{proof}
Let $\rho$ be a run from $M_0$ to $M$ in the original net.\\
We shall prove by induction over the length of $\rho$.
If $|\rho|=0$, then $\rho$ is a run in the constructed net from $M'_0$ to $A_M$.\\
Assume that the statement is true for all $|\rho|=k$.\\
Consider any run $\rho$ of length $k+1$. Let $\rho=\mu t$ where $\mu$ is a run of length $k$ and $t\in T$. Also, let $M_0\Arrow{\mu}M_1\Arrow{t}M$.\\
By induction hypothesis, either $A_{M_1}$ or $B_{M_1}$ is reachable from $M'_0$.\\
Without loss of generality, assume $A_{M_1}$ is reachable.\\
If $t\neq t_2$, then, firstly, transition $t$ is firable at $A_{M_1}$, by definition of $A_{M_1}$, since the pre-places of $t$ in have equal number of tokens as in the marking $M_1$, where it is firable.\\
Let $A_{M_1}\Arrow{t}M'$. Now, notice that $M'(p_*)=1$, $M'(p'_1)=M'(p'_*)=0$ and for all other places, we have $M'(p)=A_{M_1}(p)-F(p,t)+F(t,p)=M_1(p)-F(p,t)+F(t,p)=M(p)$. Hence, $M'=A_M$.\\
If $t=t_2$, then, again, transition $t$ is firable at $A_{M_1}$, by definition of $A_{M_1}$.\\
Let $A_{M_1}\Arrow{t}M'$. Now, notice that $M'(p'_*)=1$, $M'(p_1)=M'(p_*)=0$ and for all other places except $p'_1$, we have $M'(p)=A_{M_1}(p)-F(p,t)+F(t,p)=M_1(p)-F(p,t)+F(t,p)=M(p)$. Also, $M'(p'_1)=A_{M_1}(p_1)=M_1(p_1)=M(p_1)$. Hence, $M'=B_M$.\\
Hence proved.
\end{proof}
\begin{lemma} \label{4.26}
If marking $A_M$ or $B_M$ is reachable from $M'_0$ in the constructed net, then marking $M$ is reachable from $M_0$
\end{lemma}
\begin{proof}
We shall prove this by induction over length of the run $\rho$ from $M'_0$ to $A_M$ or $B_M$.\\
If $|\rho|=0$, then $A_M=M'_0$ is reached from $M'_0$ by $\rho$. Trivially, $M$ is reachable from $M_0$. \\
Assume that the statement is true for all runs $\mu$ of length equal to $k$.\\
Consider a run $\rho$ of length $k+1$. Let $\rho=\mu t$ where $t\in T'$ and $\mu$ is a run of length $k$. Also, without loss of generality, let $M'_0\Arrow{\mu}M'_1\Arrow{t}A_M$.\\
Then, $t$ can be either $t'_2$ or $t_i$, where $i\neq 2$.\\
If $t=t'_2$, then $M'_1(p_*)=M'_1(p_1)=0$, $M'_1(p'_*)=1$ and $M'_1(p'_1)=A_M(p_1)=M(p_1)$. Hence, $M'_1=B_{M_1}$ for some marking $M_1$.\\
By induction hypothesis, the marking $M_1$ is reachable from $M_0$ in the original net.\\
Consider firing $t_2$ at $M_1$. Clearly, transition $t_2$ is firable at $M_1$ as $t'_2$ is firable at $B_{M_1}$. Let $M_1\Arrow{t_2}M^*$. Then, $\forall p\neq p_1 M^*(p)=A_M(p)=M(p)$. Also, $M^*(p_1)=M_1(p_1)=B_{M_1}(p'_1)=A_M(p_1)=M(p_1)$. Hence, $M^*=M$. Hence, marking $M$ is reachable from $M_0$.
\end{proof}
This reduction proves that reachability is undecidable in general \Za{T}. However, as seen in Section 4.1, if the transfer arcs satisfy a structural condition, reachability is decidable.\\
Also, Deadlockfreeness in transfer petri nets (which is undecidable) is reducible to deadlockfreeness in \Za{T} by the same construction. Thus, deadlockfreeness in \Za{T} is also undecidable.\\
NOTE: The above construction also shows that coverability in \Za{IT} is undecidable, since Coverability in nets with 2 transfer arcs and 1 inhibitor arc is undecidable as shown in Section \ref{Cover}

\subsection{Hardness of Termination in \Za{IT}}
\label{appendix-62}
We use the same notations as in subsection \ref{sec:skolem}. We prove a series of statements which are finally used to prove lemma \ref{nonterminating}.

\begin{lemma} \label{lemma:nont-run}
There cannot be a non-terminating run in the forward phase. 
\end{lemma}
\begin{proof}
All the places in the net have been divided among three groups, namely $G_1$,$G_2$ and $G_3$.
We notice that, there is no transition, with a pre-place in a larger index group and post-place in a smaller or same index group. Thus, the tokens always move towards the larger index groups, and hence there can never be a non-terminating run.
\end{proof}
In any run of $N$, we call the set of transitions between $(i-1)$st firing of $t_R$ and $i$th firing of $t_R$ as $i$th forward phase, $i>1$. The first forward phase is defined as set of transitions before first firing of $t_R$. As a result of the above lemma, every run starting in $i$th forward phase either terminates or goes to $(i+1)$st forward phase.

We call place $u_{ji}$ (and transition $t_{ji}$) as incrementing, if $M_{ij}\geq 0$[i.e. post-arc from $t_{ji}$ to $u'_i$] and decrementing if $M_{ij}<0$[i.e. pre-arc from $u'_i$ to $t_{ji}$]. 
\begin{lemma} \label{lemma:invariant}
Suppose that the number of tokens in places $\{u_1,..,u_n\}$ corresponds to vector $\bar u$ before starting of $k$th forward phase and the number of tokens in all other places except $G,G'$ is zero. Then, for any marking reachable in the $k$th forward phase, $\forall 1\leq i\leq n$, we have, $$u'_i-\sum_{j}^{u_{ji}\in Decrementing} u_{ji}+\sum_{j}^{u_{ji}\in Incrementing} u_{ji}+\sum_{j}M_{ij}u_j=(M\bar u)_i$$ 
\end{lemma}
\begin{proof} 
Clearly, the equation is satisfied in the initial marking.\\
Assume that the equation is satisfied in marking $M_1$, and $M_1$\ce{->[t]}$M_2$.\\
Case 1: t is $t_j$.\\
In this case, the number of tokens in $u_j$ decreases by 1, and the number of tokens in $u_{ji}$ increases by $|M_{ij}|$. This would then have two sub cases:
\begin{itemize}
    \item $M_{ij}\geq 0$: Then, $u_{ji}$ is an incrementing place. Hence, $\sum_{j}M_{ij}u_j$ decreases by $M_{ij}$, and $\sum_{j}^{u_{ji}\in Incrementing} u_{ji}$ increases by $M_{ij}$, while all other terms remain same. Thus the summation remains constant. Hence, $LHS(M_1)=LHS(M_2)$.
    \item $M_{ij}<0$: Then, $u_{ji}$ is a decrementing place. Hence, $\sum_{j}M_{ij}u_j$ decreases by $-M_{ij}$, and $\sum_{j}^{u_{ji}\in Decrementing} u_{ji}$ decreases by $-M_{ij}$, while all other terms remain same. Thus the summation(since the terms changing have different signs in LHS) remains constant. Hence, $LHS(M_1)=LHS(M_2)$.
\end{itemize}
Case 2: t is $t_{ji}$\\
We consider two cases again :
\begin{itemize}
    \item $u_{ji}$ is an incrementing place: Then, $\sum_{j}^{u_{ji}\in Incrementing} u_{ji}$ decreases by 1 and $u'_i$ increases by 1, with all other terms remaining constant. Hence, the summation remains constant. Thus, $LHS(M_1)=LHS(M_2)$.
    \item $u_{ji}$ is a decrementing place: Then, $\sum_{j}^{u_{ji}\in Decrementing} u_{ji}$ decreases by 1 and $u'_i$ decreases by 1, with all other terms remaining constant. Hence, the summation(since the terms changing have different signs in LHS) remains constant. Thus, $LHS(M_1)=LHS(M_2)$.
\end{itemize}
Thus, in all cases, LHS remains same for both $M_1$ and $M_2$. Thus, by induction on the length of path from initial marking to any reachable marking, the given invariant holds for all reachable markings in $N$.
\end{proof}
\begin{lemma}
\label{lemma:multiplication-run}
Suppose that the number of tokens in places $\{u_1,..,u_n\}$ of marking $M_1$ corresponds to vector $\bar u$ before starting of $k$th forward phase, the number of tokens in $G = \sum_{1 \leq i\leq n} (\sum_{1\leq j\leq n}|M_{ji}|)(\bar u)_{i}$ and the number of tokens in all other places is zero. If $M\bar u\geq 0$, then there exists marking $M_2$ reachable from $M_1$ such that $M_2$ is starting of $(k+1)$st forward phase, the number of tokens in places $\{u_1,..,u_n\}$ of marking $M_2$ corresponds to vector $M \bar u$, the number of tokens in $G$ is $\sum_{1 \leq i\leq n} (\sum_{1\leq j\leq n}|M_{ji}|)(Mv_0)_{i}$ and all other places have 0 tokens. 
\end{lemma}
\begin{proof}
We obtain the marking $M_2$ from $M_1$ as follows - first, we run the transitions $t_j,1\leq j \leq n$. Next, we run the incrementing transitions among $t_{ij}$ followed by decrementing ones. As $M\bar u\geq 0$, from lemma \ref{lemma:invariant}, after the incrementing transitions are fired, we have enough transitions in $u_j', 1\leq j\leq n$ so as to finish the decrementing ones. For each firing of transition $t_{ij}$, number of tokens in $G$ is decremented by 1. So after all the incrementing and decrementing transitions $t_{ij}$ are fired, $G$ becomes empty, and hence we can fire $t_R$. The marking reached after firing $t_R$ is $M_2$. Also note that before $t_R$ is fired, from lemma \ref{lemma:invariant}, the number of tokens in $u_i'=(M\bar u)_i$. Thus, after $t_R$ is fired, number of tokens in $u_i'$ are transfered to $u_i$, which gives the fact that number of tokens in places $\{u_1,..,u_n\}$ of marking $M_2$ corresponds to vector $M \bar u$. From the construction of $G'$, before firing of $t_R$, number of tokens in $G' =\sum_{1 \leq i\leq n} (\sum_{1\leq j\leq n}|M_{ji}|)u_i'$ which are transfered to $G$ in $M_2$. Rest of the places have no tokens before firing of $t_R$ and continue to do so in $M_2$. 
\end{proof}
\begin{lemma}
\label{lemma:multiplication}
Suppose that the number of tokens in places $\{u_1,..,u_n\}$ corresponds to vector $\bar u$ before starting of $k$th forward phase, the number of tokens in $G = \sum_{1 \leq i\leq n} (\sum_{1\leq j\leq n}|M_{ji}|)(\bar u)_{i}$ and the number of tokens in all other places is zero. If $M\bar u\geq 0$, then at the end of $k$th forward phase, the number of tokens in places $\{u_1',..,u_n'\}$ corresponds to vector $M \bar u$. 
\end{lemma}
\begin{proof}
In the marking at the end of $k$th forward phase, number of tokens in $G$ is 0. As the number of tokens in $G$ at the starting of $k$ th forward phase is equal to the number of times $t_{ij}$s can fire. Thus in order to empty $G$, all the $t_{ij}$ necessarily have to fire. Thus, the number of tokens in $u_i, 1\leq i\leq n$ and $u_{ij}, 1\leq i,j\leq n$ become 0. From lemma \ref{lemma:invariant}, the current statement follows.
\end{proof}
Now we prove the original lemma.  
\begin{lemma}
There exists a non-terminating run in $N$ iff $M^k v_0 \geq 0$ for all $k \in \mathbb{N}$. 
\end{lemma}

\begin{proof}
First, we prove the reverse direction. Suppose that $M^k v_0 \geq 0$ for all $k \in \mathbb{N}$. We apply lemma \ref{lemma:multiplication-run} to prove existence of a non-terminating run. Intially at the start of 1st forward phase, our initialisation of $N$ ensures that in the current marking $M_1$, $G$ has $\sum_{1 \leq i\leq n} (\sum_{1\leq j\leq n}|M_{ji}|)(v_0)_{i}$ tokens and number of tokens in places $\{u_1,..,u_n\}$ corresponds to vector $v_0$. Thus, we reach marking $M_2$ in which $G$ has $\sum_{1 \leq i\leq n} (\sum_{1\leq j\leq n}|M_{ji}|)(Mv_0)_{i}$ tokens and number of tokens in places $\{u_1,..,u_n\}$ corresponds to vector $Mv_0$. Now, we use the lemma again to reach $M_3$ and so on. 

To prove the forward direction, assume that indeed there is a nonterminating run in $N$. From lemma \ref{lemma:nont-run}, in the nonterminating run, $t_R$ has to be fired infinitely often. 
Let the smallest $k$ such that $M^k v_0 \ngeq 0$ is $k_0$. From lemma \ref{lemma:multiplication}, before start of $k_0$th forward phase, the number of tokens in places $\{u_1,..,u_n\}$ corresponds to vector $M^{k_0-1}v_0=\bar u$ and number of tokens in $G$ is $\sum_{1 \leq i\leq n} (\sum_{1\leq j\leq n}|M_{ji}|)(\bar u)_{i}$. As $M\bar u\ngeq 0$, there exists $i$ such that $(M\bar u)_i <0$. From lemma \ref{lemma:invariant}, in any marking reachable in the current forward phase, $\sum_{j}^{u_{ji}\in Decrementing} u_{ji}>0$. As we cannot reach a marking in the current forward phase in which all of $u_{ij}$ are 0, number of tokens in $G$ is nonzero in all the markings reachable in the current forward phase. Thus, $t_R$ cannot fire, contradicting the fact that $t_R$ should be fired after $k_0$th forward phase as well, in order to have infinite firings of $t_R$. 
\end{proof}


\subsection{Summary}
\label{sec:app-summary}
A comprehensive extension of table \ref{summary} is presented here. 

\begin{centering}
\begin{tabular}{|c|c|}
\hline \Za{I} & Refer to table \ref{summary} \\
\hline \Za{T} & Refer to table \ref{summary} \\
\hline \Za{R} & Decidable - subsumed by \Za{IR} \\

\hline \Za{IR} & Refer to table \ref{summary} \\
\hline \Za{IT} & Refer to table \ref{summary} \\
\hline \Za{RT} & \term,\cover : Decidable (\cite{finkel}). \reach,\dlf : Undecidable - subsumes \Za{T} \\

\hline \Zb{R}{I} & Refer to table \ref{summary} \\
\hline \Zb{R}{T} & \term,\cover : Decidable (\cite{finkel}). \reach,\dlf : Undecidable - subsumes \Za{T} \\
\hline \Zb{R}{R} & Equivalent to \Zc{R} \\

\hline \Zb{I}{I} & Undecidable : equivalent to \Zc{I} \\
\hline \Zb{I}{T} & Undecidable : subsumes \Zc{I} \\
\hline \Zb{I}{R} & Undecidable : subsumes \Zc{I} \\

\hline \Zb{T}{I} &  Refer to table \ref{summary}\\
\hline \Zb{T}{T} & Equivalent to \Zc{T}\\
\hline \Zb{T}{R} & \term,\cover : Decidable (\cite{finkel}). \reach,\dlf : Undecidable - subsumes \Zc{T} \\

\hline \Zb{R}{IR} & Refer to table \ref{summary} \\
\hline \Zb{R}{IT} & \term : Positivity hard (Thm.\ref{thm:skolem-hard}), Others : Undecidable - from \Zb{R}{I}\\
\hline \Zb{R}{RT} & \term,\cover : Decidable (\cite{finkel}). \reach,\dlf : Undecidable - subsumes \Zc{R}\\

\hline \Zb{I}{IR} & Undecidable : subsumes \Zc{I}\\
\hline \Zb{I}{IT} & Undecidable : subsumes \Zc{I}\\
\hline \Zb{I}{RT} & Undecidable : subsumes \Zc{I}\\

\hline \Zb{T}{IR} &  \term : Positivity hard (Thm.\ref{thm:skolem-hard}), Others : Undecidable - from \Zb{T}{I}\\
\hline \Zb{T}{IT} & \term : Positivity hard (Thm.\ref{thm:skolem-hard}), Others : Undecidable - from \Zb{T}{I}\\
\hline \Zb{T}{RT} & \term,\cover : Decidable (\cite{finkel}). \reach,\dlf : Undecidable - subsumes \Zc{T} \\

\hline \Za{IRcT} & Refer to table \ref{summary} \\

\hline 
\end{tabular}
\end{centering}

The current status of the relative expressiveness and decidability of termination, coverability, reachability and deadlockfreeness can be visualised for various classes of nets as follows -

\begin{center}
\begin{tikzpicture}
\draw (4,9.5) rectangle node{\Zb{R}{IR}} (6,10.5);
\draw (6,7.5) rectangle node{\Zb{R}{I}} (8,8.5);
\draw (2,7.5) rectangle node{\Za{IR}} (4,8.5);
\draw (6,5.5) rectangle node{\Zb{1R}{I}} (8,6.5);
\draw (2,3.5) rectangle node{\Za{I}} (4,4.5);
\draw[-latex,thick] (5,9.5) -- (7,8.5);
\draw[-latex,thick] (5,9.5) -- (3,8.5);
\draw[-latex,thick] (7,7.5) -- (7,6.5);
\draw[latex-latex,thick] (3,7.5) -- (3,4.5);
\draw[-latex,thick] (7,5.5) -- (3,4.5);
\draw[dotted] (1.5,9) node[above, xshift=0.5cm]{C,R,D} rectangle (4.5,3);
\draw[dotted] (1,11) node[above,  xshift=0.3cm]{T} rectangle (9,2.5);
\node at (7.5,6.65) {C?,$\neg$R};
\node at (5.6,8.65) {$\neg$C,$\neg$R};
\node at (5.6,10.65) {$\neg$C,$\neg$R};
\draw (10,6) rectangle (16,2.5);
\node at (12.2,5.5) {X: Problem X is decidable};
\node at (12.4,4.8) {$\neg$X: Problem X is undecidable};
\node at (12,4.1) {X?: Problem X is open};
\draw[-] (10,3.7) -- (16,3.7);
\node at (13,3.4) {C: Coverability D: Deadlockfreeness};
\node at (12.7,2.7) {T: Termination R: Reachability};
\end{tikzpicture}
\begin{tikzpicture}
\draw (1.5,9.5) rectangle node{\Zb{RcT}{IRcT}} (4.5,10.5);
\draw (4,7.5) rectangle node{\Zb{RcT}{I}} (6,8.5);
\draw (0,7.5) rectangle node{\Za{IcT}} (2,8.5);
\draw (0,3.5) rectangle node{\Za{I}} (2,4.5);
\draw (7,7.5) rectangle node{\Za{IT}} (9,8.5);
\draw (7,3.5) rectangle node{\Za{T}} (9,4.5);
\draw[-latex,thick] (3,9.5) -- (5,8.5);
\draw[-latex,thick] (3,9.5) -- (1,8.5);
\draw[-latex,thick] (1,7.5) -- (1,4.5);
\draw[-latex,thick] (8,7.5) -- (1,4.5);
\draw[-latex,thick] (8,7.5) -- (8,4.5);
\draw[dotted] (-0.5,9) node[above, xshift=0.5cm]{C,R,D} rectangle (2.5,3);
\draw[dotted] (-1,11) node[above,  xshift=0.3cm]{T} rectangle (6.5,2.5);
\node at (3.6,8.65) {$\neg$C,$\neg$R};
\node at (3.6,10.65) {$\neg$C,$\neg$R};
\node at (8.6,8.65) {T:S};
\node at (8.8,4.65) {T,C,$\neg$R};
\draw (10,6) rectangle (16,2.5);
\node at (12.2,5.5) {X: Problem X is decidable};
\node at (12.4,4.8) {$\neg$X: Problem X is undecidable};
\node at (12.8,4.1) {X:S: Problem X is Positivity Hard};
\draw[-] (10,3.7) -- (16,3.7);
\node at (13,3.4) {C: Coverability D: Deadlockfreeness};
\node at (12.7,2.7) {T: Termination R: Reachability};
\end{tikzpicture}
\end{center}

\end{document}